\documentclass[journal,twoside]{IEEEtran}
\usepackage{subfigure}
\usepackage{setspace}
\usepackage{amsmath}
\usepackage{amssymb}
\usepackage{amsfonts}
\usepackage{amscd}
\usepackage{mathrsfs}
\usepackage[final]{graphicx}
\usepackage{graphicx}
\usepackage{psfrag}
\usepackage{color}
\usepackage{url}

\newtheorem{theorem}{Theorem}
\newtheorem{lemma}{Lemma}
\newtheorem{definition}{Definition}
\newtheorem{proposition}{Proposition}

\newtheorem{example}{Example}

\begin{document}

\title{Coding for Two-User SISO and MIMO Multiple Access Channels}
\author{J. Harshan and B. Sundar Rajan, Senior Member, IEEE
\thanks{This work was supported partly through grants to B.Sundar Rajan by the DRDO-IISc program on Advanced Research in Mathematical Engineering. Part of the content of this paper is in the proceedings of IEEE International Symposium on Information theory (ISIT 2008). Some parts of this paper is submitted to IEEE ISIT 2009 to be held at Seoul, Korea.
The authors are with the Department of Electrical Communication Engineering, Indian Institute of Science, Bangalore-560012, India. Email:\{harshan,bsrajan\}@ece.iisc.ernet.in.}}
\maketitle

\begin{abstract}
Constellation Constrained (CC) capacity regions of a two-user SISO Gaussian Multiple Access Channel (GMAC) with finite complex input alphabets and continuous output are computed in this paper. When both the users employ the same code alphabet, it is well known that an appropriate rotation between the alphabets provides unique decodability to the receiver. For such a set-up, a metric is proposed to compute the angle(s) of rotation between the alphabets such that the CC capacity region is maximally enlarged. Subsequently, code pairs based on Trellis Coded Modulation (TCM) are designed for the two-user GMAC with $M$-PSK and $M$-PAM alphabet pairs for arbitrary values of $M$ and it is proved that, for certain angles of rotation, Ungerboeck labelling on the trellis of each user maximizes the guaranteed squared Euclidean distance of the \textit{sum trellis}. Hence, such a labelling scheme can be used systematically to construct trellis code pairs for a two-user GMAC to achieve sum rates close to the sum capacity of the channel. More importantly, it is shown for the first time that ML decoding complexity at the destination is significantly reduced when $M$-PAM alphabet pairs are employed with \textit{almost} no loss in the sum capacity.\\
\indent A two-user Multiple Input Multiple Output (MIMO) fading MAC with $N_{t}$ antennas at both the users and a single antenna at the destination has also been considered with the assumption that the destination has the perfect knowledge of channel state information and the two users have the perfect knowledge of only the phase components of their channels. For such a set-up, two distinct classes of Space Time Block Code (STBC) pairs derived from the well known class of real orthogonal designs are proposed such that the STBC pairs are information lossless and have low ML decoding complexity. 
\end{abstract}


\section{Introduction and Preliminaries}
\label{sec1}
Capacity regions of a two-user Gaussian Multiple Access Channel (GMAC) with continuous input alphabets and continuous output is well known \cite{ThC}-\cite{Li}. Such a model assumes the users to employ Gaussian code alphabets and the additive noise is assumed to be Gaussian distributed. Though, capacity regions of such a channel provides insights in to the achievable rate pairs ($R_{1}$, $R_{2}$) in an information theoretic sense, it fails to provide information on the achievable rate pairs when we consider finitary  restrictions on the input alphabets and analyze some real world practical signal constellations like $M$-QAM and $M$-PSK signal sets for some positive integer, $M$. Gaussian multiple access channels with finite complex input alphabets and continuous output was first considered in \cite{FTL} with the assumption of random phase offsets in the channel from every user to the destination. For such a channel model, Constellation Constrained (CC) sum capacity has been computed for some well known alphabets such as $M$-PSK and $M$-QAM in \cite{FTL} and trellis codes have also been proposed in \cite{AuE} and \cite{FAR}. Note that the assumption of random phase offsets in the channel naturally provides Uniquely Decodable (UD) property to the receiver when all the users use the same alphabet.\\
\indent Subsequently, a $K$-user GMAC model with no random phase offsets in the channel has been considered in \cite{WCA} and codes based on Trellis Coded Modulation (TCM) \cite{Ub} have been proposed wherein the UD property at the destination is achieved by employing distinct alphabets for all the users. In particular, an alphabet of size $KM$ (example: $KM$-PSK or $KM$- QAM) is chosen and it is appropriately partitioned in to $K$ groups such that every user uses one of the groups as its code alphabet. Towards designing trellis codes, the authors of \cite{WCA} only propose steps to choose the labellings on the edges of the trellises of all the users, but do not provide explicit labellings on the individual trellises.\\
\indent In this paper, a two-user GMAC with finite complex input alphabets and continuous output is considered without the assumption of random phase offsets in the channel (it is shown that the assumption of random phase offsets in the channel leads to loss in the CC sum capacity) and the impact of the rotation between the alphabets of the two-users on capacity regions is investigated. Code pairs based on TCM are also proposed such that sum rates close to the CC sum capacity can be achieved for some classes of alphabet pairs. Throughout the paper, the mutual information value for a GMAC when the symbols from the input alphabets are chosen with uniform distribution is referred as the Constellation Constrained (CC) capacity of the GMAC \cite{Eb}. Henceforth, unless specified, (i) input alphabet refers to a finite complex alphabet and a GMAC refers to a Gaussian MAC with finite complex input alphabets and continuous output alphabet and (ii) capacity (capacity region) refers to the CC capacity (capacity region). Throughout the paper, the terms alphabet and signal set are used interchangeably. \\
\indent The idea of rotation between the alphabets of the two users is also exploited to construct Space Time Block Code (STBC) pairs with low ML decoding complexity for a two-user MIMO (Multiple Input Multiple Output) fading MAC. For a background on space-time coding for MIMO-MAC, we refer the reader to \cite{GaB, MaB}. Till date, we are not aware of any work which address the design of STBC pairs with low ML decoding complexity for a two-user MIMO-MAC. Note that STBC pairs with minimum ML decoding complexity have been well studied in the literature for collocated MIMO channels \cite{TJC, Xl, ZaR, KaR} and relay channels \cite{YiK, HaR2} as well. The contributions of the paper may be summarized as below: 
\begin{itemize}
\item For a two-user GMAC, when both the users employ identical alphabet, it is well known that an appropriate rotation between the alphabets can guarantee UD property (See Definition \ref{ud_definition}) to the receiver \cite{FTL}. For such a setup, we identify that the primary problem is to compute the angle(s) of rotation between the alphabets such that the capacity region is maximally enlarged. A metric to compute the angle(s) of rotation is proposed which provides maximum enlargement of the capacity region (See Theorem \ref{thm}). Through simulations, such angles of rotation are presented for some well known alphabets such as $M$-PSK, $M$-QAM etc for some values of $M$ at some fixed SNR values (See Table \ref{rotation_table1}).
\item For a two-user GMAC, code pairs based on TCM are designed with $M$-PSK and $M$-PAM alphabet pairs to achieve sum rates close to the CC sum capacity of the channel. In particular, trellis codes are explicitly designed for each user by exploiting the structure of the sum alphabets of $M$-PSK and $M$-PAM alphabet pairs.
\item For each $i = 1, 2$, if User-$i$ employs the trellis $T_{i}$ labelled with the symbols of the signal set $\mathcal{S}_{i}$, it is clear that the destination sees the sum trellis, $T_{sum}$ (See Definition \ref{def_sum_trellis}) labelled with the symbols of the sum alphabet, $\mathcal{S}_{sum}$ (See Section \ref{sec2_subsec1}) in an equivalent SISO (Single Input Single Output) AWGN channel. Recall that, for a SISO AWGN channel, Ungerboeck labelling on the trellis maximizes the guaranteed minimum squared Euclidean distance, $d^{2}_{g, min}$ of the trellis and hence such a labelling scheme has become a systematic method of generating trellis codes to go close to the capacity \cite{Ub}. However, when TCM based trellis codes are designed for a two-user GMAC, it is not clear if the two users can distributively achieve Ungerboeck labelling on the sum trellis through the trellises $T_{1}$ and $T_{2}$. In other words, it is not known whether Ungerboeck labelling on $T_{1}$ and $T_{2}$ using $\mathcal{S}_{1}$ and $\mathcal{S}_{2}$ respectively induces an Ungerboeck labelling on $T_{sum}$ using $\mathcal{S}_{sum}$. In this paper, it is analytically proved that, for the class of symmetric $M$-PSK signal sets, when the relative angle is $\frac{\pi}{M}$, Ungerboeck labelling on the trellis of each user induces an Ungerboeck labelling on $T_{sum}$ which in-turn maximizes the $d^{2}_{g, min}$ of the $T_{sum}$ (See Theorem \ref{ungerboeck_theorem}). Hence, such a labelling scheme can be used as a systematic method of generating trellis code pairs for a two-user GMAC to go close to the sum capacity. An example for an alphabet pair is presented using which a non-Ungerboeck labelling on the trellis of each user maximizes the $d^{2}_{g, min}$ of the $T_{sum}$. (See Section \ref{sec3_subsec3}, Example \ref{counter_example_labelling}).
\item Trellis code pairs are also designed in this paper with $M$-PAM signal sets for a two-user GMAC (See Section \ref{sec3_subsec4}). For such signal sets, it is shown that the relative angle of rotation that maximally enlarges the CC capacity region is $\frac{\pi}{2}$ $\forall M$ and for all values of SNR. Note that the above structure on $M$-PAM alphabet pairs keep the two users orthogonal to each other and hence trellis codes designed for SISO AWGN channel with $M$-PAM alphabets are applicable in this set-up. Hence the ML decoding complexity is significantly reduced when trellis codes with $M$-PAM signal sets are employed. Through simulations, it is shown that, for a particular SNR, the sum capacity of $4$-PAM signal sets (when used with a relative rotation of $\frac{\pi}{2}$) and QPSK signal sets (with appropriate angles of rotation provided in \cite{HaR}) are \textit{almost} the same and hence unlike in a SISO AWGN channel there is no loss in the sum capacity by using $4$-PAM alphabets over QPSK signal sets in a two user GMAC. Similar observations are also presented for Gaussian code alphabets.
\item A two-user MIMO fading MAC with $N_{t}$ antennas at both the users and a single antenna at the destination is considered in this paper. The destination is assumed to have the perfect knowledge of Channel State Information (CSI) whereas the two users are assumed to have the perfect knowledge of only the phase components of their channels to the destination. For such a set-up, two classes of Space Time Block Code (STBC) pairs are introduced such that the code pairs are (i) information lossless and (ii) have reduced ML decoding complexity for all values of $N_{t}$. To the best of our knowledge, this is the first paper that addresses construction of STBC pairs for MIMO-MAC with low ML decoding complexity as well as information losslessness property. (See Section \ref{sec4}).
\item The first class of STBC pairs are from a class of complex designs called Separable Orthogonal Designs (SOD) (See Definition \ref{def_sod}) which in-turn are constructed using the well known class of Real Orthogonal Designs (ROD) \cite{TJC, Xl}. It is shown that STBC pairs generated from SODs are (i) information lossless for all values of $N_{t}$ and (ii) are two-group ML decodable \cite{KaR2}. (See Section \ref{sec4_subsec1}).
\item The second class of STBC pairs are generated straight from RODs wherein certain restrictions on input alphabet pairs are imposed. Such a class of STBC pairs are shown to be information lossless for large values of $N_{t}$. However, for smaller values of $N_{t}$, the loss in the sum capacity is shown to be marginal. Importantly, the proposed code pairs also have the single symbol ML decodable property as they are generated straight from RODs. Simulation results are presented which show that STBC pairs from RODs perform better than those from SODs in terms of total Bit Error Rate (BER). (See Section \ref{sec4_subsec2}).
\end{itemize}
\textit{Notations:} For a random variable $X$ which takes value from the set $\mathcal{S}$, we assume some ordering of its elements and use $X(i)$ to represent the $i$-th element of $\mathcal{S}$. i.e. $X(i)$ represents a realization of the random variable $X$. Cardinality of the set $\mathcal{S}$ is denoted by $|\mathcal{S}|$. Absolute value of a complex number $x$ is denoted by $|x|$ and $E \left[x\right]$ denotes the expectation of the random variable $x$. A circularly symmetric complex Gaussian random vector, $x$ with mean $\mu$ and covariance matrix $\Gamma$ is denoted by $x \sim \mathcal{CN} \left(\mu, \Gamma \right)$. Also, the set of all real and complex numbers are denoted by $\mathbb{R}$ and $\mathbb{C}$ respectively. For $a, b \in \mathbb{C}$, distance between between $a$ and $b$ is denoted by $d(a, b)$ whereas the line segment connecting $a$ and $b$ is denoted by $l(a, b)$.

The remaining content of the paper is organized as follows: In Section \ref{sec2}, we present CC capacity regions of a two-user GMAC with finite alphabet pairs and provide details on computing the angles of rotation between the alphabets such that the CC capacity region is maximally enlarged. In Section \ref{sec3}, we discuss details on designing TCM schemes for a two-user GMAC with $M$-PSK and $M$-PAM as input alphabet pairs $\forall M$. In Section \ref{sec4}, a two-user MIMO-MAC model is introduced and two distinct classes of low ML decoding complexity STBC pairs are presented. Section \ref{sec5} constitutes conclusion and some directions for possible future work.
\begin{figure}
\centering
\includegraphics[width=2.3in]{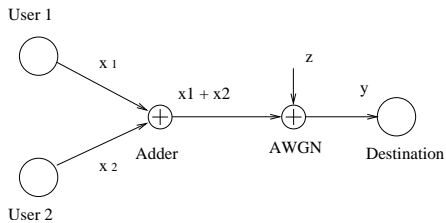}
\caption{Two-user Gaussian MAC model} 
\label{gmac_model}
\end{figure}
\section{Two-user GMAC - signal model and constellation constrained capacity regions}
\label{sec2}
The model of a two-user Gaussian MAC is as shown in Fig. \ref{gmac_model} consists of two users that need to convey information to a single destination. It is assumed that User-1 and User-2 communicate to the destination at the same time and in the same frequency band. Symbol level synchronization is assumed at the destination. The two users are equipped with alphabets $\mathcal{S}_{1}$ and $\mathcal{S}_{2}$ of size $N_{1}$ and $N_{2}$ respectively. When User-1 and User-2 transmit symbols $x_{1} \in \mathcal{S}_{1}$ and $x_{2} \in \mathcal{S}_{2}$ respectively, the destination receives a symbol $y$ given by,
\begin{equation}
\label{gmac_model}
y = x_{1} + x_{2} + z ~~ \mbox{where }  z \sim \mathcal{CN} \left(0, \sigma^{2} \right).
\end{equation}
We compute the mutual information values $I(x_{2} : y)$ and $I(x_{1} : y ~|~ x_{2})$ when the symbols $x_{1}$ and $x_{2}$ are assumed to take values form $\mathcal{S}_{1}$ and $\mathcal{S}_{2}$ with uniform distribution. By symmetry, $I(x_{1} : y)$ and $I(x_{2} : y ~|~ x_{1})$ can also be computed. Considering $x_{1} + z$ as the additive noise, $I(x_{2} : y)$ and $I(x_{1} : y ~|~ x_{2})$ can be computed and are presented in \eqref{mi} and \eqref{mi1} (shown at the top of the next page).
\begin{figure*}
\begin{equation}
\label{mi}
I(x_{2} : y) = \mbox{log}_{2}(N_{2}) - \frac{1}{N_{1}N_{2}}\sum_{k_{1} = 0}^{N_{1} - 1}\sum_{k_{2} = 0}^{N_{2} - 1}E\left[\mbox{log}_{2}\left[ \frac{\sum_{i_{1} = 0}^{N_{1} - 1}\sum_{i_{2} = 0}^{N_{2} - 1} \mbox{exp}\left(- |x_{1}(k_{1}) + x_{2}(k_{2}) - x_{1}(i_{1}) - x_{2}(i_{2}) + z|^{2}/\sigma^{2}\right)}{\sum_{i_{1} = 0}^{N_{1} - 1} \mbox{exp}\left(- |x_{1}(k_{1}) - x_{1}(i_{1})+ z|^{2}/\sigma^{2}\right)}\right] \right].
\end{equation}
\begin{equation}
\label{mi1}
I(x_{1} : y ~|~x_{2} ) = \mbox{log}_{2}(N_{1}) - \frac{1}{N_{1}}\sum_{k_{1} = 0}^{N_{1} - 1}E\left[\mbox{log}_{2}\left[ \frac{\sum_{i_{1} = 0}^{N_{1} - 1} \mbox{exp}\left(- |x_{1}(k_{1}) -  x_{1}(i_{1}) + z|^{2}/\sigma^{2}\right)}{\mbox{exp}\left(- |z|^{2}/\sigma^{2}\right)}\right] \right].
\end{equation}
\hrule
\end{figure*}
Therefore, using \eqref{mi} and \eqref{mi1}, the sum mutual information for both the users is $I(x_{2} : y)$ + $I(x_{1} : y~|~x_{2})$. Using Fano's inequality, it is straightforward to prove that rates (in bits per channel use) more than the above mutual information values cannot be achieved. Hence, the capacity region is as given below \cite{ThC},
$$R_{1} ~< ~ I(x_{1} : y ~|~ x_{2}),$$
$$R_{2} ~< ~ I(x_{2} : y ~|~ x_{1}) \mbox{ and }$$
\begin{equation}
\label{capacity_region}
R_{1}  + R_{2} ~ < ~ I(x_{1}, x_{2} : y ) = I(x_{1} : y ~|  ~x_{2}) + I(x_{2} : y).
\end{equation}
Note that $I(x_{2} : y)$ and $I(x_{1} : y ~|~ x_{2})$ (similarly $I(x_{1} : y)$ and $I(x_{2} : y ~|~ x_{1})$) are only upper bounds on the achievable rate pairs since coding schemes achieving rate pairs close to the edges of the CC capacity region are yet to be identified.
\subsection{Uniquely decodable alphabet pairs for GMAC}
\label{sec2_subsec1}
In this subsection, we formally define a UD alphabet pair. Given two alphabets $\mathcal{S}_{1}$ and $\mathcal{S}_{2}$, we denote the sum alphabet of $\mathcal{S}_{1}$ and $\mathcal{S}_{2}$ by $\mathcal{S}_{sum}$ defined as $\mathcal{S}_{sum} = \left\lbrace x_{1} + x_{2} ~ | ~ \forall ~ x_{1} \in \mathcal{S}_{1}, x_{2} \in \mathcal{S}_{2}\right\rbrace$. The adder channel in the two-user GMAC (as shown in  Fig. \ref{gmac_model}) can be viewed as a mapping $\phi$ given by $\phi : \mathcal{S}_1 \times \mathcal{S}_2 \longrightarrow \mathcal{S}_{sum} \mbox{ where } \phi((x_{1}, x_{2})) = x_{1} + x_{2}$.

\begin{definition}
\label{ud_definition}
An alphabet pair ($\mathcal{S}_{1}$, $\mathcal{S}_{2}$) is said to be uniquely decodable if the mapping $\phi$
is one-one.
\end{definition}

\indent Example for a UD alphabet pair is $\mathcal{S}_{1} = \left\lbrace 1, -1 \right\rbrace$ and $\mathcal{S}_{2} = \left\lbrace i, -i \right\rbrace$. An example for a non-UD alphabet pair is given by $\mathcal{S}_{1} = \mathcal{S}_{2} = \left\lbrace 1, -1 \right\rbrace$. Note that if $S_1$ and $S_2$ have more than one element common, then the pair $(\mathcal{S}_1,\mathcal{S}_2)$ is necessarily non-UD. However, not having more than one common signal point is not sufficient for a pair to be UD, as exemplified by the pair $S_1=\{1, \omega, \omega^2  \}$ and $S_2=\{-1, 1+\omega, 1+\omega^2 \}$ where $\omega$ is a cube root of unity.

It is clear that uncoded multi-user communication with non-UD alphabet pair results in ambiguity while performing joint decoding for the symbols of both users at the destination. In order to circumvent this ambiguity, the two users can jointly construct code pairs ($\mathcal{C}_{1}, \mathcal{C}_{2}$) (codes constructed by adding redundancy across time) over the non-UD alphabet pair so that the codewords of both users can be uniquely decoded. Note that, there will be a loss in the rate of transmission (in other words, there will be an expansion of the bandwidth) by adopting such schemes. Therefore, for band-limited multiuser Gaussian channels, coding across time is forbidden to achieve the UD property and hence, the use of UD alphabets is essential.
\subsection{Capacity maximizing alphabet pairs from rotations}
\label{sec2_subsec2}
For a GMAC with $\mathcal{S}_{1} = \mathcal{S}_{2}$, it is clear that if one of the users employ an appropriate rotated version of the alphabet used by the other, then UD property is attainable. Moving one step further, we consider the problem of finding the optimal angle(s) of rotation between the alphabet pair such that the capacity region is maximally enlarged for a given of SNR.

\indent For a given alphabet $\mathcal{S}_{1}$, let $\mathcal{S}_{2}$ denote the set of symbols obtained by rotating all the symbols of $\mathcal{S}_{1}$ by $\theta$ degrees. From \eqref{capacity_region}, the capacity region is determined by the mutual information values $I(x_{1} : y ~|~ x_{2})$, $I(x_{2} : y ~|~ x_{1})$ and $I(x_{2} : y)$ (or $I(x_{1} : y)$). Note that, the terms $I(x_{1} : y ~|~ x_{2})$ and $I(x_{2} : y ~|~ x_{1})$ are functions of the Distance Distribution (DD) of $\mathcal{S}_{1}$ and $\mathcal{S}_{2}$ respectively. Since, we start with a known $\mathcal{S}_{1}$ and $S_{2} = e^{i\theta}S_{1}$, the DD of $\mathcal{S}_{1}$ and $\mathcal{S}_{2}$ are the same. Hence $I(x_{1} : y ~|~ x_{2})$ and $I(x_{2} : y ~|~ x_{1})$ are independent of $\theta$. However, from \eqref{mi}, the term $I(x_{2} : y)$ is a function of the DD of $\mathcal{S}_{sum}$ and the DD of $\mathcal{S}_{sum}$ changes with $\theta$ and hence the term $I(x_{2} : y)$ is a function of $\theta$.

In the following theorem, we provide a criterion to choose the value of $\theta$ such that $I(x_{2} : y)$ is maximized which in-turn maximally enlarges the capacity region in \eqref{capacity_region}.
\begin{theorem}
\label{thm}
Let ($\mathcal{S}_{1}$, $\mathcal{S}_{2}$) be an alphabet pair such that $\mathcal{S}_{2} = e^{i \theta}\mathcal{S}_{1}$ and $| \mathcal{S}_{1} | = N$. The mutual information value, $I(x_{2} : y)$ in \eqref{mi} is maximized by choosing the angle of rotation, $\theta^{*}$ such that $\theta^{*} = \mbox{arg} \min_{\theta \in (0, 2 \pi)} M(\theta)$ where $M(\theta)$ is given in \eqref{thmeq}.
\begin{figure*}
\begin{equation}
\label{thmeq}
M(\theta) = \mbox{arg} \min_{\theta \in (0, 2 \pi)} \sum_{k_{1} = 0}^{N - 1}\sum_{k_{2} = 0}^{N - 1}\mbox{log}_{2}\left[ \sum_{i_{1} = 0}^{N - 1}\sum_{i_{2} = 0}^{N - 1} \mbox{exp}\left(- |x_{1}(k_{1}) - x_{1}(i_{1}) + e^{i \theta} ( x_{1}(k_{2})  - x_{1}(i_{2})) |^{2}/4\sigma^{2}\right)\right] .
\end{equation}
\end{figure*}

\end{theorem}
\begin{proof}
Since $N_{1}$ = $N_{2}$ = $N$ is fixed, $\mbox{arg} \max_{\theta \in (0, 2 \pi)} I(x_{2} : y) = \mbox{arg} \min_{\theta \in (0, 2 \pi)} I^{'}(x_{2} : y)$ where $I^{'}(x_{2} : y)$ is given in \eqref{prfeq1}.
\begin{figure*}
\begin{equation}
\label{prfeq1}
I^{'}(x_{2} : y) = \sum_{k_{1} = 0}^{N_{1} - 1}\sum_{k_{2} = 0}^{N_{2} - 1}E\left[\mbox{log}_{2}\left[ \frac{\sum_{i_{1} = 0}^{N_{1} - 1}\sum_{i_{2} = 0}^{N_{2} - 1} \mbox{exp}\left(- |x_{1}(k_{1}) + x_{2}(k_{2}) - x_{1}(i_{1}) - x_{2}(i_{2}) + z|^{2}/\sigma^{2}\right)}{\sum_{i_{1} = 0}^{N_{1} - 1} \mbox{exp}\left(- |x_{1}(k_{1}) - x_{1}(i_{1})+ z|^{2}/\sigma^{2}\right)}\right] \right].
\end{equation}
\end{figure*}
\noindent Since the denominator term inside the logarithm of $I^{'}(x_{2} : y)$ is independent of $\theta$, $\mbox{arg} \min_{\theta \in (0, 2 \pi)} I^{'}(x_{2} : y) = \mbox{arg} \min_{\theta \in (0, 2 \pi)} I^{''}(x_{2} : y)$ where $I^{''}(x_{2} : y)$ is given in \eqref{prfeq2}.
\begin{figure*}
\begin{equation}
\label{prfeq2}
I^{''}(x_{2} : y) = \sum_{k_{1} = 0}^{N_{1} - 1}\sum_{k_{2} = 0}^{N_{2} - 1}\underbrace{E\left[\mbox{log}_{2}\left[ \sum_{i_{1} = 0}^{N_{1} - 1}\sum_{i_{2} = 0}^{N_{2} - 1} \mbox{exp}\left(- |x_{1}(k_{1}) + x_{2}(k_{2}) - x_{1}(i_{1}) - x_{2}(i_{2}) + z|^{2}/\sigma^{2}\right)\right] \right]}_{\lambda(k_{1}, k_{2})}.
\end{equation}
\hrule
\end{figure*}
Applying Jensen's inequality on the individual terms, $\lambda(k_{1}, k_{2})$ of $I^{''}(x_{2} : y)$ and replacing each term of the form $x_{2}(.)$ by $e^{i \theta}x_{1}(.)$, we have $I^{''}(x_{2} : y) \leq M(\theta)$. Hence, instead of finding $\theta$ which minimizes $I^{''}(x_{2} : y)$, we propose to find $\theta^{*}$ which minimizes $M(\theta)$, an upper bound on $I^{''}(x_{2} : y)$.
\end{proof}

\indent Note that every individual term, $\lambda(k_{1}, k_{2})$ of $I^{''}(x_{2} : y)$ is an expectation of a non linear function of the random variable $z$ and the closed form expressions of $\lambda(k_{1}, k_{2}) ~\forall~ k_{1}, k_{2}$ are not available. Therefore, in the above theorem, we propose to find $\theta^{*}$ which minimizes $M(\theta)$, an upper bound on $I^{''}(x_{2} : y)$ instead of $I^{''}(x_{2} : y)$ itself. Note that the values of $\theta$ obtained by minimizing $I^{''}(x_{2} : y)$ can potentially provide larger capacity regions than those obtained using $\theta^{*}$.

\indent Since $M(\theta)$ depends on the DD of $\mathcal{S}_{sum}$, $\theta^{*}$ depends on the average transmit power per channel use $P_{1}$ (for User-1) and $P_{2}$ (for User-2) of the alphabets $\mathcal{S}_{1}$ and $\mathcal{S}_{2}$ respectively. Note that, though the results of the Theorem \ref{thm} applies only to alphabets such that $\mathcal{S}_{2} = e^{i \theta}\mathcal{S}_{1}$, its extension to the case when $\mathcal{S}_{2} \neq e^{i \theta}\mathcal{S}_{1}$ is straightforward.

\subsection{Optimal rotations for some known alphabets}
\label{sec2_subsec3}
\begin{center}
\begin{table*}
\caption{Two-tuples ($a$, $b$) for $M$-PSK and $M$-QAM alphabets for some $M$ : $a$ - $\theta^{*}$. $b$ - multiplicity of $\theta^{*}$  (from SNR = -2 db to SNR = 16 db) }
\begin{center}
\begin{tabular}{|c|c|c|c|c|c|c|c|c|c|c|}
\hline SNR in db & BPSK & QPSK & 8-QAM & 8-PSK & 16-PSK\\
\hline -2 & (90, 1) & (45.0, 1) & (90, 1) & (22.5, 1) & (01.43, 1)\\
\hline 0 & (90, 1) & (45.0, 1) & (90, 1) & (22.5, 1) & (09.12, 2)\\
\hline 2 & (90, 1) & (45.0, 1) & (90, 1) & (22.5, 1) & (21.37, 1)\\
\hline 4 & (90, 1) & (45.0, 1) & (90, 1) & (22.5, 1) & (11.25, 1)\\
\hline 6 & (90, 1) & (45.0, 1) & (90, 1) & (22.5, 1) & (11.25, 1)\\
\hline 8 & (90, 1) & (35.1, 2) & (110.12, 1) & (22.5, 1) & (11.25, 1)\\
\hline 10 & (90, 118) & (58.1, 1) & (61.62, 1) & (18.5, 1) & (11.25, 1)\\
\hline 12 & (90, 775) & (30.8, 1) & (119.25, 1) & (16.0, 1) & (09.50, 1)\\
\hline 14 & (90, 1269) & (30.5, 1) & (118.75, 1) & (15.3, 1) & (10.56, 1)\\
\hline 16 & (90, 1609) & (30.3, 2) & (118.0, 1) & (15.1, 1) & (08.31, 2)\\
\hline 
\end{tabular} 
\end{center}
\label{rotation_table1}
\end{table*}
\end{center}

In this subsection, for a given alphabet $\mathcal{S}_{1}$ and for a given value of $P_{1}$ = $P_{2}$, through simulations, using the metric proposed in Theorem 1, we find angle(s) of rotation, $\theta^{*}$ (in degrees) that results in a DD of the sum alphabet which maximizes $I(x_{2} : y)$. For the simulation results, additive noise, $z$ is assumed to have unit variance per dimension. i.e. $\sigma^{2} = 2$. The values of $\theta^{*}$ are calculated by varying the relative angle of rotation from 0 to 180 in steps of 0.0625 degrees. In Table \ref{rotation_table1}, for various values of $P_{1}/\sigma^{2}$ = SNR, values of $\theta^{*}$ are presented for some well known alphabets such as $M$-QAM, $M$-PSK for $M$ = 4, 8 and 16. Against every signal set, a two-tuple $(a, b)$ is presented where $a$ denotes $\theta^{*}$ and $b$ represents the multiplicity of $\theta^{*}$ since for some SNR values, there could be more than one value of $\theta^{*}$ that minimizes $M(\theta)$ (Example : QPSK at SNR = 8 db, 16-PSK at SNR = 16 db). In general, if $\theta^{*}$ is calculated by varying the angle of rotation with different intervals, then the optimal $\theta^{*}$ and the multiplicity of the optimal $\theta^{*}$ will also change. When there are multiple values of $\theta^{*}$ for a signal set, only one of them is provided in the table. Among the several angles, the one presented in the table reduces the complexity at the transmitters compared to the rest of the angles (Example : for BPSK, 90 degrees is chosen over other angles of rotation at SNR = 10 db).  However, when there is not much difference in the complexity among the several values of $\theta^{*}$, we present the one with the least value (Example : for QPSK at SNR = 8 db, 16 db).

\subsubsection{Capacity regions of a GMAC with ${\cal S}_1=$BPSK}
\label{sec2_subsec3_subsubsec1}
\indent In Fig. \ref{bpsk_rr}, capacity regions using BPSK alphabet pair with optimal rotation and without rotation are given at SNR  = -2 db and 2 db. Capacity regions of a GMAC with Gaussian alphabets are also given in Fig. \ref{bpsk_rr} at -2 db and 2 db. The plot shows that, for a given SNR, capacity region of the BPSK alphabet pair is contained inside the capacity region of the Gaussian code alphabet. Note that, with rotation, both users can transmit at rates equal to SISO AWGN channel capacity with BPSK alphabet simultaneously. This is because $\theta^{*} = 90$ degrees (at all SNR values) which makes $\mathcal{S}_{1}$ and $\mathcal{S}_{2}$ orthogonal. Hence both users can achieve the rates close to $I(x_{1} : y ~|~ x_{2})$ and $I(x_{2} : y ~|~ x_{1})$ respectively at all SNR values. From Table \ref{rotation_table1}, note that there are several angles apart from 90 degrees which minimizes $M(\theta)$ even though they do not provide orthogonality to the users. The reason being; for BPSK constellation, the SNR values of 10 db and higher are enough to make the additive noise at the destination is negligible and hence a non zero angle of rotation (not necessarily 90) is sufficient for both users to communicate 1 bit each. In general, multiple optimal angles exist for any alphabet at values of SNR beyond which the capacity region saturates.
\begin{figure}
\centering
\includegraphics[width=3in]{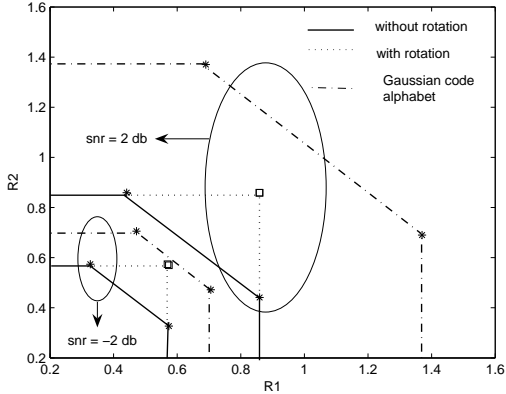}
\caption{Capacity regions of BPSK alphabet pair with optimal rotation and without rotation at SNR = -2 and 2 db} 
\label{bpsk_rr}
\end{figure}
\subsubsection{Capacity regions of a GMAC with ${\cal S}_1=$QPSK}
\label{sec2_subsec3_subsubsec2}
\indent Capacity regions for QPSK alphabet pair is shown with optimal rotation and without rotation at different SNR values in Fig. \ref{qpsk_rr_1}. It is to be observed that rotation provides enlarged capacity region from the SNR value of 2 db onwards. However, at SNR  = 0 db, capacity regions with optimal rotation and without rotation coincides. The percentage increase in $I(x_{2} : y)$ ranges from 4.3 percent at 2 db to 100 percent asymptotically. At SNR  = 6 db, capacity region of a GMAC with Gaussian alphabets is also provided and it can be observed that the capacity region with Gaussian alphabets contains the capacity region of a GMAC with QPSK alphabet.
\begin{figure}
\centering
\includegraphics[width=3in]{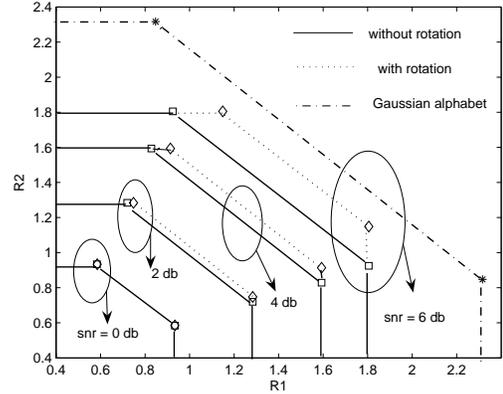}
\caption{Capacity regions of QPSK alphabet pair with optimal rotation and without rotation at SNR = 0, 2, 4 and 6 db} 
\label{qpsk_rr_1}
\end{figure}
\subsection{Two-User GMAC with random phase-offsets}
\label{sec2_subsec5}
In this subsection, capacity regions of a GMAC computed using the channel model presented in \eqref{gmac_model} are compared with those of a GMAC when random phase offsets are introduced in the channel. The GMAC channel model with random offsets has been considered in \cite{FTL}, wherein the constellation constrained capacity of the resulting sum alphabet has been computed in an AWGN channel. For such a setup, it is clear that the problem of designing UD alphabet pairs is completely avoided. However, there will be a loss in the CC sum capacity since the relative angle between the alphabets is a random variable which can also take values other than $\theta^{*}$. Since $I(x_{2} : y)$ is the only term which is variant to rotations, we have plotted $I(x_{2} : y)$ at different SNR values with and without random offsets for BPSK and QPSK alphabet pairs in Fig. \ref{bpsk_rand_offsets} and Fig. \ref{qpsk_rand_offsets} respectively. For the case with no random offsets, values of $\theta^{*}$ presented in Table \ref{rotation_table1} are used to maximise $I(x_{2} : y)$.
\begin{figure}
\centering
\includegraphics[width=3in]{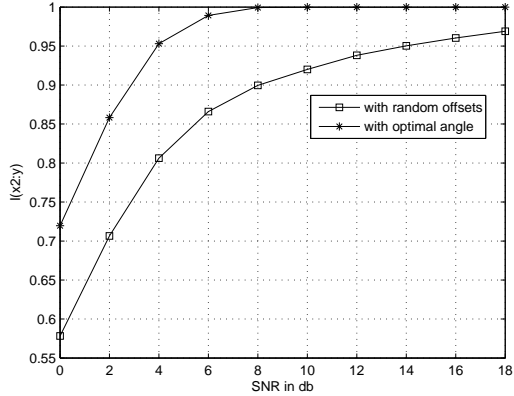}
\caption{$I(x_{2} : y)$ for BPSK alphabet pair with (i) random offsets and (ii) without random offsets} 
\label{bpsk_rand_offsets}
\end{figure}
\begin{figure}
\centering
\includegraphics[width=3in]{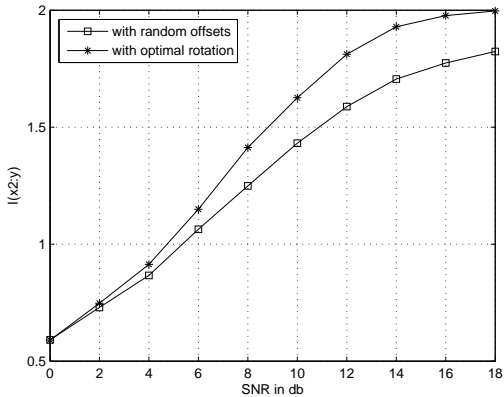}
\caption{$I(x_{2} : y)$ for QPSK alphabet pair with (i) random offsets and (ii) without random offsets} 
\label{qpsk_rand_offsets}
\end{figure}
\section{Trellis Coded Modulation (TCM) for a two-user GMAC}
\label{sec3}
In this section, we design code pairs based on TCM for a two-user GMAC to achieve sum rates, $R_{1}+ R_{2}$ close to the CC sum capacity of the channel (given in \eqref{capacity_region}). Towards that direction, the following proposition indicates the need for coding by both the users.
\begin{proposition}
\label{prop1}
For a two-user GMAC, if one of the users is equipped with a TCM based trellis code, then the other user also needs to employ a TCM based trellis code to achieve sum rates close to the sum capacity. 
\end{proposition} 
\begin{proof}
If User-1 employs TCM and User-2 performs uncoded transmission, then for any trellis, $T_{1}$ chosen by User-1, since User-2 has a trivial trellis, the sum trellis, $T_{sum}$ (See Definition \ref{def_sum_trellis}) will have parallel paths. For a trellis with parallel paths, it is well known that the minimum accumulated squared Euclidean distance of the trellis is equal to the minimum squared Euclidean distance between the points labelled along the parallel paths. Hence, in order to achieve larger values of minimum accumulated Euclidean distance in the sum trellis, both the users need to employ trellis codes with larger number of states so that sum rates close to the CC sum capacity can be achieved
\end{proof}

\indent For each $i = 1, 2$, let User-$i$ be equipped with a convolutional encoder $C_{i}$ with $m_{i}$ input bits and $m_{i} + 1$ output bits. Throughout the section, we consider convolutional codes which add only 1-bit redundancy. Let the $m_{i} + 1$ output bits of $C_{i}$ take values from a 2-dimensional signal set $\mathcal{S}_{i}$ such that $|\mathcal{S}_{i}| = 2^{m_{i} + 1}$. Henceforth, the codes $\mathcal{C}_{1}$ (set of codewords generated from $C_{1}$) and $\mathcal{C}_{2}$ (set of codewords generated from $C_{2}$) are represented by trellises $T_{1}$ and $T_{2}$ respectively. The sum trellis, $T_{sum}$ for the trellis pair $\left(T_{1}, T_{2}\right)$ is introduced in the following definition:  
\begin{definition}
\label{def_sum_trellis}
Let $T_{1}$ and $T_{2}$ represent two trellises with $n+1$ stages having the state complexity profiles $\left\lbrace q_{1,0}, q_{1,1}, \cdots q_{1,n} \right\rbrace$ and $\left\lbrace q_{2,0}, q_{2, 1}, \cdots q_{2, n} \right\rbrace$ respectively. Let $E^{a}_{1,i}$ and $E^{b}_{2,i}$ respectively denote the edge sets originating from the state $(a)$ of $T_{1}$ and the state $(b)$ of $T_{2}$ in the $i$-th stage where $1 \leq a \leq q_{1,i}$ and $1 \leq b \leq q_{2,i}$. Let the edge sets $E^{a}_{1,i}$ and $E^{b}_{2,i}$ be labelled with the symbols of the sets $\mathcal{X}^{a}_{i}$ and $\mathcal{Y}^{b}_{i}$ respectively. For the above trellis pair, the sum trellis, $T_{sum}$ is a $n + 1$ stage trellis such that  
\begin{itemize}
\item The state complexity profile is, 
\begin{equation*}
\left\lbrace q_{1,0}q_{2,0}, ~q_{1,1}q_{2,1}, \cdots ~q_{1,n}q_{2,n} \right\rbrace
\end{equation*}
where a particular state in the $i$-th stage is denoted by $\left(a, b\right)$ such that $1 \leq a \leq q_{1,i}$ and $1 \leq b \leq q_{2,i}$.
\item The edge set originating from the state $\left(a, b\right)$ in the $i$-th stage is given by $E^{(a,b)}_{i} = E^{a}_{1,i} \times E^{b}_{1,i}$. In particular, if $2^{m_{1}}$ and $2^{m_{2}}$ edges originate from state $(a)$ and state $(b)$ of $T_{1}$ and $T_{2}$ in the $i$-th stage respectively, then $2^{m_{1} + m_{2}}$ edges originate from the state $\left(a, b\right)$ in the $i$-th stage.
\item The edges of the set $E^{(a,b)}_{i}$ are labelled with the symbols of the set $\mathcal{X}^{a}_{i} + \mathcal{Y}^{b}_{i}$.
\end{itemize}
\end{definition}
\begin{example}
For the trellis pair (shown in Fig. \ref{trp}) labelled with elements of $\mathcal{S}_{1}$ and $\mathcal{S}_{2}$ (shown in Fig. \ref{uncoded_alphabets}), the sum trellis $T_{sum}$ is as shown in Fig. \ref{str} which is labelled with the elements of $\mathcal{S}_{sum}$ (shown in Fig. \ref{expanded_alphabet}).
\end{example}
\begin{figure}
\centering
\includegraphics[width=2.5in]{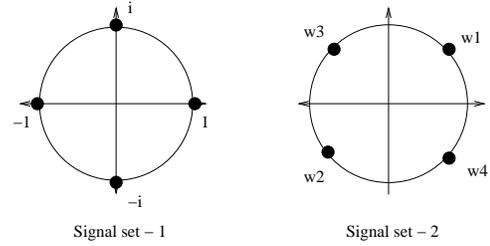}
\caption{Code alphabets used by User-1 and User-2} 
\label{uncoded_alphabets}
\end{figure}
\begin{figure}
\centering
\includegraphics[width=2.5in]{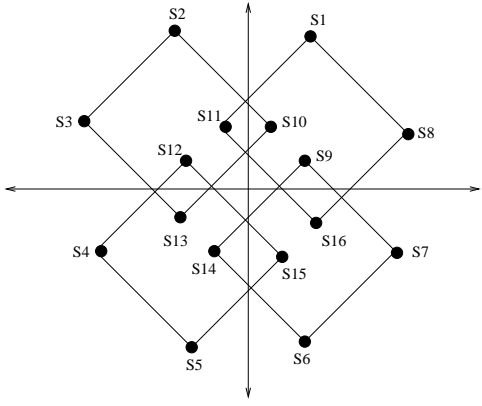}
\caption{Sum alphabet, $\mathcal{S}_{sum}$ for the signal sets presented in Fig. \ref{uncoded_alphabets}.}  
\label{expanded_alphabet}
\end{figure}
\begin{figure}
\centering
\includegraphics[width=2.5in]{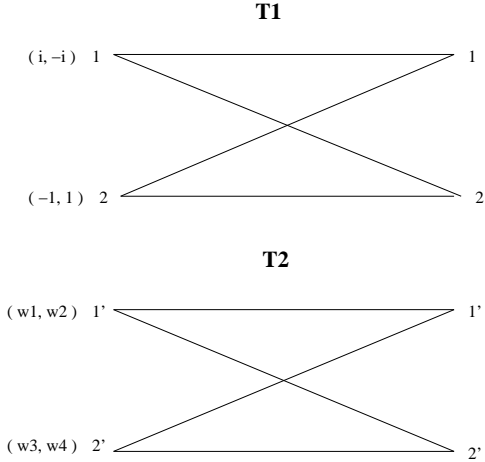}
\caption{Two state trellises of User-1 ($\mbox{T}_{1}$) and User -2 ($\mbox{T}_{2}$).} 
\label{trp}
\end{figure}
\begin{figure}
\centering
\includegraphics[width=3in]{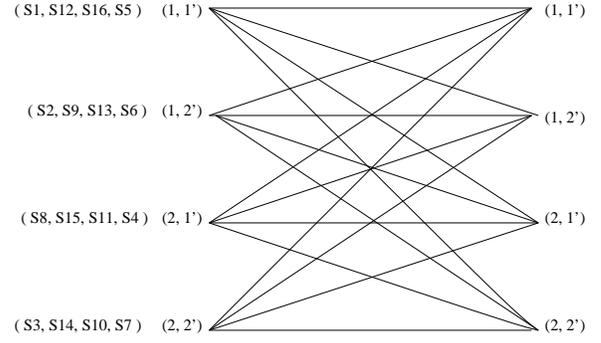}
\caption{Sum trellis, $\mbox{T}_{sum}$ of trellises $\mbox{T}_{1}$ and $\mbox{T}_{2}$ presented in Fig. \ref{trp}.} 
\label{str}
\end{figure}

\indent We assume that the destination performs joint decoding of the symbols of User-1 and User-2 by decoding for a sequence over $\mathcal{S}_{sum}$ on the sum trellis, $T_{sum}$. For the trellis pair $(T_{1}, T_{2})$ and the alphabet pair $(\mathcal{S}_{1}, \mathcal{S}_{2})$, the destination views an equivalent SISO AWGN channel with a virtual source equipped with the trellis, $T_{sum}$ labelled with the elements of $\mathcal{S}_{sum}$. Recall that for a SISO AWGN channel, if the source is equipped with a trellis, $T$ and an alphabet $\mathcal{S}$, the following design rules have been suggested in \cite{Ub} to construct \textit{good} trellis codes.
\begin{itemize}
\item All the symbols of $\mathcal{S}$ should occur with equal frequency and with some amount of regularity.
\item Transitions originating from the same state (or joining the same state) must be labelled with subsets of $\mathcal{S}$ whose minimum Euclidean distance is maximized. 
\end{itemize}

\indent Due to the existence of an equivalent AWGN channel in the GMAC set-up, the sum trellis, $T_{sum}$ has to be labelled with the elements of $\mathcal{S}_{sum}$ satisfying the above design rules. However, from the design point of view, such a labelling rule can be obtained on $T_{sum}$ only through the pairs $(T_{1}, T_{2})$ and $(\mathcal{S}_{1}, \mathcal{S}_{2})$. Hence, in this section, we propose labelling rules on $T_{1}$ and $T_{2}$ using $\mathcal{S}_{1}$ and $\mathcal{S}_{2}$ respectively such that $T_{sum}$ is labelled with the elements of $\mathcal{S}_{sum}$ as per Ungerboeck rules. The problem statement has been elaborately explained below.\\
\indent Since the number of input bits to $C_{i}$ is $m_{i}$, there are $2^{m_{i}}$ edges diverging from (or converging to; henceforth, we only refer to diverging edges) each state of $T_{i}$. Also, as there is only one bit redundancy to be added by the code and $|S_{i}| = 2^{m_{i} + 1}$, the edges diverging from each state have to be labelled with the elements of a subset of $\mathcal{S}_{i}$ of size $2^{m_{i}}$. Therefore, for each $i$, $\mathcal{S}_{i}$ has to be partitioned in to two sets $\mathcal{S}^{1}_{i}$ and $\mathcal{S}^{2}_{i}$ and the diverging edges from each state of $T_{i}$ have to be labelled with the elements of either $\mathcal{S}^{1}_{i}$ or $\mathcal{S}^{2}_{i}$. From the definition of a sum trellis, there are $2^{m_{1} + m_{2}}$ edges diverging from each state of $T_{sum}$ which gets labelled with the elements of one of the following sets,
\begin{equation*}
\mathcal{A} = \left\lbrace \mathcal{S}^{i}_{1} + \mathcal{S}^{j}_{2} ~ \forall ~ i, j = 1, 2 \right\rbrace. 
\end{equation*}
As per the Ungerboeck design rules, the transitions originating from the same state of $T_{sum}$ must be assigned symbols that are separated by largest minimum distance. Hence, the problem addressed in this section is to find an optimal partitioning of $\mathcal{S}_{i}$ in to two sets $\mathcal{S}^{1}_{i}$ and  $\mathcal{S}^{2}_{i}$ of equal cardinality such that the minimum squared Euclidean distance of each one of the sets in $\mathcal{A}$ is maximized. However, since $d_{min}$ of the sets in $\mathcal{A}$ can potentially be different, we try to find an optimal partitioning such that the minimum of the $d_{min}$ values of the sets in $\mathcal{A}$ is maximized.\\
\indent In particular, we try to propose solution to the above problem when $\mathcal{S}_{1}$ and $\mathcal{S}_{2}$ are symmetric $M$-PSK signal sets such that $\mathcal{S}_{2} = e^{i\theta}\mathcal{S}_{1}$ for any $\theta$ satisfying $0 < \theta < \frac{2\pi}{M}$ where $M = 2^{r}$ for $r \geq 1$. Note that $\mathcal{S}_{2} = e^{i\theta}\mathcal{S}_{1}$ implies $m_{1} = m_{2}$. As a first step towards solving the above problem, in the following subsection, we study the structure of the sum alphabet, $\mathcal{S}_{sum}$ of two $M$-PSK signal sets which are of the form $\mathcal{S}_{2} = e^{i\theta}\mathcal{S}_{1}$ for any $\theta$ satisfying $0 < \theta < \frac{2\pi}{M}$.
\begin{figure}
\centering
\includegraphics[width=3in]{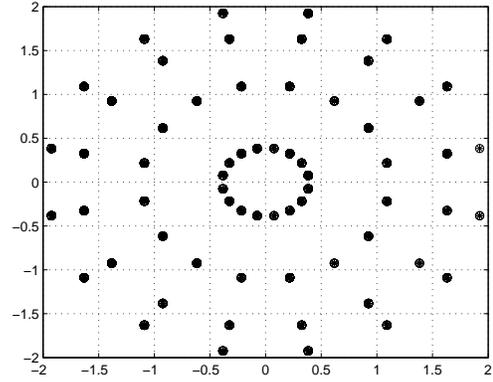}
\caption{The structure of $\mathcal{S}_{sum}$ when $\mathcal{S}_{1}$ = 8-PSK and $\theta = \frac{\pi}{8}$} 
\label{sum_alphabet_pi_M}
\end{figure}
\begin{figure}
\centering
\includegraphics[width=3in]{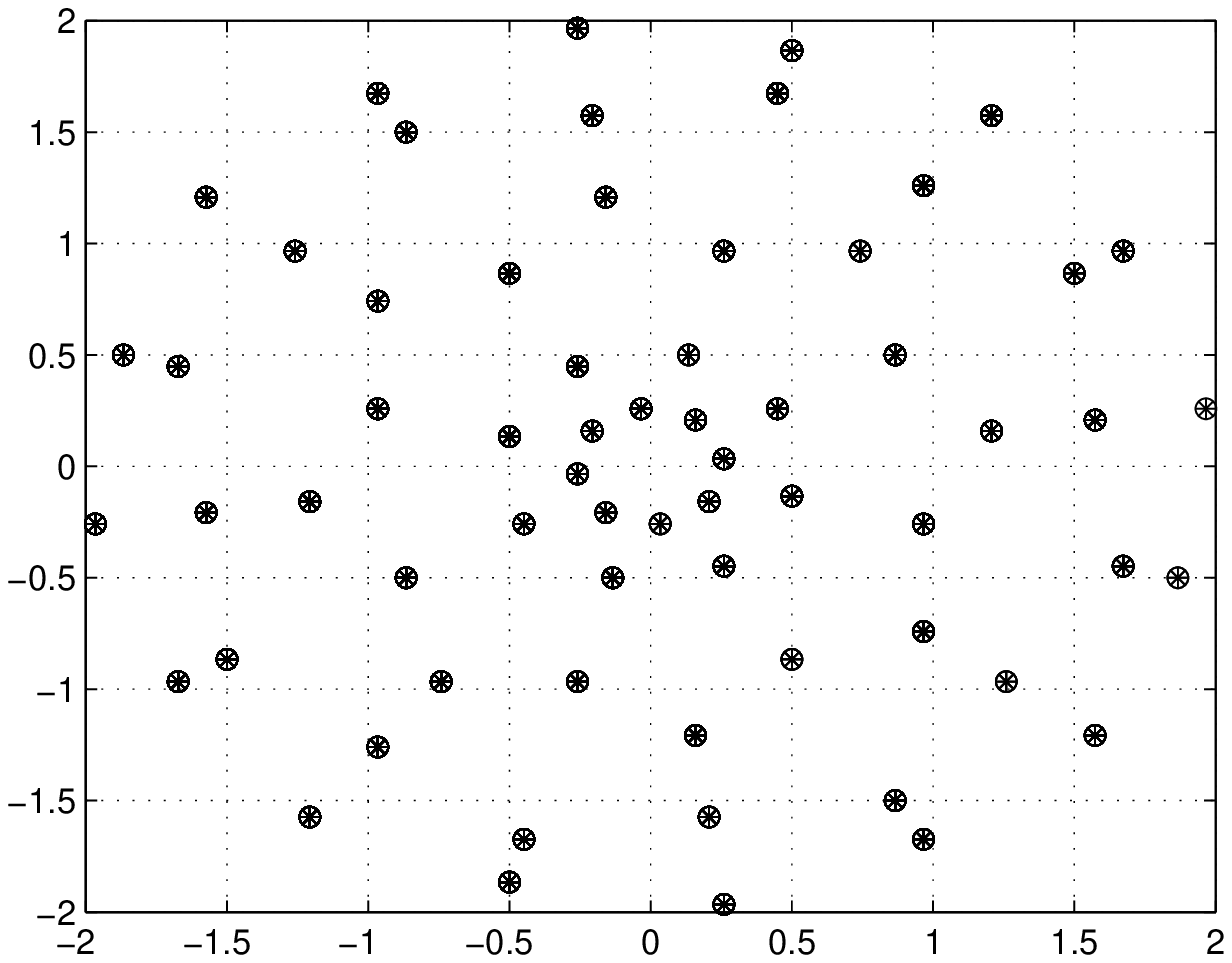}
\caption{The structure of $\mathcal{S}_{sum}$ when $\mathcal{S}_{1}$ = 8-PSK and $\theta = \frac{\pi}{12}$} 
\label{sum_alphabet_not_pi_M}
\end{figure}
\subsection{Structure of $\mathcal{S}_{sum}$ when $\mathcal{S}_{1}$ is a $M$-PSK signal set}
\label{sec3_subsec1}
Let $\mathcal{S}_{1}$ and $\mathcal{S}_{2}$ represent two symmetric $M$-PSK signal sets such that $\mathcal{S}_{2} = e^{i\theta}\mathcal{S}_{1}$ where $0 < \theta < \frac{2\pi}{M}$. Let $x(n)$ and $x'(n)$ denote the points $e^{\frac{i2\pi n}{M}}$ and $e^{\frac{i2\pi n}{M}}e^{i\theta}$ of $\mathcal{S}_{1}$ and $\mathcal{S}_{2}$ respectively for $ 0 \leq n \leq M - 1$. The sum alphabet, $\mathcal{S}_{sum}$ of $\mathcal{S}_{1}$ and $\mathcal{S}_{2}$ is given by,
\begin{equation*}
\mathcal{S}_{sum} = \mathcal{S}_{1} + \mathcal{S}_{2} = \left\lbrace x(n) + x'(n') ~|~ \forall ~ 0 \leq n, n' \leq M - 1 \right\rbrace.
\end{equation*}
Alternatively, $\mathcal{S}_{sum}$ can be written as given in \eqref{sum_alph_alt} (at the top of the next page)
\begin{figure*}
\begin{equation}
\label{sum_alph_alt}
\mathcal{S}_{sum} =  \left\lbrace x(n) + x'(n + m), x(n) + x'(n - m - 1) ~|~ \forall ~ 0 \leq n \leq M - 1 \mbox{ and } 0 \leq m \leq M/2 - 1 \right\rbrace
\end{equation}
\hrule
\end{figure*}
where $x(n) + x'(n + m) = e^{\frac{i2\pi n}{M}} + e^{i\{\frac{2\pi (n + m)}{M} + \theta\}}$ and $x(n) + x'(n - m - 1) = e^{\frac{i2\pi n}{M}} + e^{i\{\frac{2\pi (n - m)}{M} - \theta\}}$ such that $x'(-p) = x'(M-p)$ for any $0 \leq p \leq M-1$. The phase components of the points $x(n) + x'(n + m)$ and $x(n) + x'(n - m - 1)$ are given by $\frac{2\pi n}{M} + \frac{\pi m}{M} +  \frac{\theta}{2}$ and  $\frac{2\pi n}{M} - \frac{\pi (m + 1)}{M} +  \frac{\theta}{2}$ respectively. For a fixed $m$, the set of points of the form $x(n) + x'(n + m)$ lie on a circle of radius $2\mbox{cos}(\frac{\pi m}{M} + \frac{\theta}{2})$ and let that circle be denoted by $O^{m}$. Similarly, for a fixed $m$, the collection of points of the form $x(n) + x'(n - m - 1)$ lie on a circle of radius $2\mbox{cos}(\frac{\pi (m + 1)}{M} - \frac{\theta}{2})$ and the circle is denoted by $I^{m}$. Therefore, $\mathcal{S}_{sum}$ takes the structure of $M$ concentric PSK signal sets (as shown in Fig. \ref{sum_alphabet_not_pi_M}). The structure of $\mathcal{S}_{sum}$ for a 8-PSK signal set is shown in Fig. \ref{sum_alphabet_pi_M} and Fig. \ref{sum_alphabet_not_pi_M} for $\theta = \frac{\pi}{8}$ and $\theta = \frac{\pi}{12}$ respectively. The set containing the radii of the $M$ circles is given by

{\footnotesize
\begin{equation*}
\mathcal{R} = \left\lbrace 2\mbox{cos}(\frac{\pi m}{M} + \frac{\theta}{2}), ~2\mbox{cos}(\frac{\pi (m + 1)}{M} - \frac{\theta}{2}) ~|~ \forall ~ 0 \leq m \leq \frac{M}{2} - 1 \right\rbrace. 
\end{equation*}
}

\noindent Henceforth, throughout the section, $r(O^{m})$ and $r(I^{m})$ denote the radius of the circle $O^{m}$ and $I^{m}$ respectively. Note that when $\theta =  \pi/M$, $r(O^{m}) = r(I^{m})$ and hence $\mathcal{S}_{sum}$ has the structure of $\frac{M}{2}$ concentric PSK signal sets (as shown in Fig. \ref{sum_alphabet_pi_M}). For $\theta \in \left(0, \frac{\pi}{M} \right)$, it can be verified that the structure of $\mathcal{S}_{1} + e^{i\theta}\mathcal{S}_{1}$ is identical to that of $\mathcal{S}_{1} + e^{i(\theta + \frac{\pi}{M})}\mathcal{S}_{1}$ and henceforth, we consider values of $\theta$ such that $0 < \theta < \frac{\pi}{M}$. Note that the elements of $\mathcal{R}$ satisfies the following relation,

{\footnotesize
\begin{equation*}
r(I^{M/2 - 1}) \leq r(O^{M/2 - 1}) \leq r(I^{M/2 - 2}) \leq r(O^{M/2 - 2}) \leq ~ \cdots ~ \leq r(O^{0}).
\end{equation*}
}

\noindent For the elements of $\mathcal{R}$, the following three propositions can be proved using standard trigonometric identities :
\begin{proposition}
\label{prop2}
The sequence $\left\lbrace r(O^{k}) - r(I^{k}) \right\rbrace$ from $k = 0$ to $M/2 -1$ is an increasing sequence.
\end{proposition}
\begin{proposition}
\label{prop3}
The sequence $\left\lbrace r(O^{k}) - r(I^{k + 1}) \right\rbrace$ from $k = 0$ to $M/2 - 2$ is an increasing sequence.
\end{proposition}
\begin{proposition}
\label{prop4}
The sequence $\left\lbrace r(I^{k}) - r(O^{k + 1}) \right\rbrace$ from $k = 0$ to $M/2 - 2$ is an increasing sequence.
\end{proposition}
\begin{proposition}
\label{prop5}
Using the phase information of each point in $\mathcal{S}_{sum}$, the following observations can be made:
\begin{enumerate}

\item The angular separation between the two points $x(n) + x'(n + m)$ and $x(n') + x'(n' + m)$ on $O^{m}$ is $\frac{2 \pi (n-n')}{M}$ for all $m = 0 \mbox{ to } M/2 - 1$. Similarly, the angular separation between the two points $x(n) + x'(n - m - 1)$ and $x(n') + x'(n' - m - 1)$ on $I^{m}$ is $\frac{2 \pi (n-n')}{M}$ for all $m = 0 \mbox{ to } M/2 - 1$.
\item The angular separation between the point, $x(n) + x'(n + m)$ on $O^{m}$ and the point $x(n') + x'(n' - m - 1)$ on $I^{m}$ is $\frac{2 \pi (n-n')}{M} + \frac{\pi (2m + 1)}{M}$ for all $m = 0 \mbox{ to } M/2 - 1$.
\item The angular separation between the point $x(n) + x'(n + m)$ on $O^{m}$ and the point $x(n') + x'(n' - (m-1) - 1)$ on $I^{m-1}$ is $\frac{2 \pi (n-n')}{M} + \frac{\pi (2m)}{M}$ for all $m = 1 \mbox{ to } M/2 - 1$.
\item The angular separation between the point $x(n) + x'(n - m - 1)$ on $I^{m}$ and the point $x(n') + x'(n' + (m-1))$ on $O^{m-1}$ is $\frac{2 \pi (n-n')}{M} - \frac{\pi (2m)}{M}$ for all $m = 1 \mbox{ to } M/2 - 1$.\\
\end{enumerate}
\end{proposition}

\indent In the next subsection, first, we partition each $\mathcal{S}_{i}$ in to two groups using Ungerboeck rules and then, exploiting the structure of $\mathcal{S}_{sum}$, we compute the minimum Euclidean distance, $d_{min}$ of each one of the sets in $\mathcal{A}$.
\subsection{Structure of each one of the sets in $\mathcal{A}$ induced by the Ungerboeck partitioning on $\mathcal{S}_{1}$ and $\mathcal{S}_{2}$}
\label{sec3_subsec2}
For each $i = 1, 2$, let $\mathcal{S}_{i}$ be partitioned in to two sets of equal size using Ungerboeck rules which results in two sets denoted by $\mathcal{S}^{e}_{i}$ and $\mathcal{S}^{o}_{i}$ such that $d_{min}$ of $\mathcal{S}^{e}_{i}$ and $\mathcal{S}^{o}_{i}$ is maximized. Since the number of sets resulting from the partition is two, the minimum angular separation, $\phi_{min}$ between the points in each set is $\frac{4\pi}{M}$. The two sets of $\mathcal{S}_{i}$ are of the form,
\begin{equation*}
\mathcal{S}^{e}_i = \left\lbrace x(n) ~|~ n = 2m ~ \mbox{ for }  0 \leq m \leq M/2 - 1\right\rbrace \mbox{ and }
\end{equation*}
\begin{equation*}
\mathcal{S}^{o}_i = \left\lbrace x(n) ~|~ n = 2m + 1 ~ \mbox{ for } 0 \leq m \leq M/2 - 1\right\rbrace.
\end{equation*}
 It is clear that the sets $\mathcal{S}^{e}_1 + \mathcal{S}^{o}_2$, $\mathcal{S}^{e}_1 + \mathcal{S}^{e}_2$, $\mathcal{S}^{o}_1 + \mathcal{S}^{o}_2$ and $\mathcal{S}^{o}_1 + \mathcal{S}^{e}_2$ $\in \mathcal{A}$ form a partition of $\mathcal{S}_{sum}$. In the rest of this subsection, we obtain the $d_{min}$ values of the above sets. Throughout this section, the set $\mathcal{S}^{\alpha}_1 + \mathcal{S}^{\beta}_2$ and its minimum distance are denoted by $S^{\alpha\beta}_{sum}$ and $d^{\alpha\beta}_{min}$ respectively $\forall \alpha, \beta = e$ and $o$. Without loss of generality, only the structure of $S^{eo}_{sum}$ and $S^{ee}_{sum}$ are studied since the structure of $S^{oe}_{sum}$ and $S^{oo}_{sum}$ are identical to that of $S^{eo}_{sum}$ and $S^{ee}_{sum}$ respectively.\\
\subsubsection{Calculation of $d_{min}$ of $\mathcal{S}^{e}_1 + \mathcal{S}^{o}_2$}
\begin{figure*}
\begin{equation}
\label{subset_1}
\mathcal{S}^{e}_1 + \mathcal{S}^{o}_2 = \left\lbrace x(n) + x'(n') ~ | ~ n = 2m \mbox{ and } n' = 2m' + 1 ~\forall ~ m, m' = 0 \mbox{ to } M/2 -1 \right\rbrace. 
\end{equation}
\hrule
\end{figure*}
The elements of $\mathcal{S}^{e}_1 + \mathcal{S}^{o}_2$ (as given in \eqref{subset_1}) are of the form $x(n) + x'(n + m)$ and $x(n) + x'(n - m - 1)$ where $n$ takes even numbers while $n + m$ and $n - m - 1$ take odd numbers. When $m$ is odd, note that $n + m$ is odd and $n - m - 1$ is even and hence $\mathcal{S}^{e}_1 + \mathcal{S}^{o}_2$ will have no points on $I^{m}$ and $M/2$ points on $O^{m}$. Similarly, when $m$ is even, $\mathcal{S}^{e}_1 + \mathcal{S}^{o}_2$ will have no points on $O^{m}$ and $M/2$ points on $I^{m}$. Therefore, $\mathcal{S}^{e}_1 + \mathcal{S}^{o}_2$ will have points on the following set of circles $\left\lbrace O^{M/2 - 1}, I^{M/2 - 2}, O^{M/2-3}, \cdots , O^{1}, I^{0} \right\rbrace$.

\indent Since $n$ takes only even values, using observation 1) of Proposition \ref{prop5}, $\phi_{min}$ between the points of $\mathcal{S}^{e}_1 + \mathcal{S}^{o}_2$ on any circle is $\frac{4\pi}{M}$. Hence the points of $\mathcal{S}^{e}_1 + \mathcal{S}^{o}_2$ are maximally separated on every circle. Also, $\phi_{min}$ between the points placed on consecutive circles has to be calculated. From observation 3) of Proposition \ref{prop5}, $\phi_{min}$ between the points in $O^{q}$ and $I^{q-1}$ is $\frac{2\pi}{M}$ for all $q = 1$ to $M/2 - 1$. Similarly $\phi_{min}$ between the points in $I^{q}$ and $O^{q-1}$ is $0$ for all $q = 2$ to $M/2 - 2$. It can be verified that the structure of $\mathcal{S}^{o}_1 + \mathcal{S}^{e}_2$ is identical to that of $\mathcal{S}^{e}_1 + \mathcal{S}^{o}_2$.
\begin{proposition}
\label{prop6}
$r(I^{q-1})$ and $r(O^{q})$ satisfy the following inequality for all $q = 1$ to $M/2 - 1$,
\begin{equation*}
d(r(I^{q-1}),  r(O^{q})e^{i\frac{2\pi}{M}}) \geq d(r(O^{q}), r(O^{q})e^{i\frac{2\pi}{M}})
\end{equation*}
\end{proposition}
\begin{proof}
See Appendix \ref{proof_prop_6}.
\end{proof}
\begin{proposition}
\label{prop7}
For $M \geq 8$, $r(O^{q})$ satisfy the following inequality for all $q = 1$ to $M/2 - 3$,
\begin{equation*}
d(r(O^{q}), r(O^{q})e^{i\frac{2\pi}{M}}) \geq 2r(O^{M/2-1})\mbox{sin}(\frac{2\pi}{M}).
\end{equation*}
\end{proposition}
\begin{proof}
See Appendix \ref{proof_prop_7}
\end{proof}
\begin{lemma}
For $0 < \theta < \frac{\pi}{M}$, the minimum distance of the sets $\mathcal{S}^{eo}_{sum}$ and $\mathcal{S}^{oe}_{sum}$ are given by
{\small
\begin{equation}
\label{min_dist_set_1}
d^{oe}_{min} = d^{eo}_{min} = \mbox{min}\left(d(r(I^{q-1}),  r(O^{q})e^{i\frac{2\pi}{M}}), 4\mbox{sin}\left( \frac{\pi}{M} - \frac{\theta}{2}\right) \right).
\end{equation}
}
\end{lemma}
\begin{proof}
See Appendix \ref{proof_lemma_1}.\\
\end{proof}
\subsubsection{Calculation of $d_{min}$ of $\mathcal{S}^{e}_1 + \mathcal{S}^{e}_2$}
The set $\mathcal{S}^{e}_1 + \mathcal{S}^{e}_2$ is as given in \eqref{subset_2}.
\begin{figure*}
\begin{equation}
\label{subset_2}
\mathcal{S}^{e}_1 + \mathcal{S}^{e}_2 = \left\lbrace x(n) + x'(n') ~ | ~ n = 2m \mbox{ and } n' = 2m' ~\forall ~ m, m' = 0 \mbox{ to } M/2 -1 \right\rbrace. 
\end{equation}
\hrule
\end{figure*}
The elements of $\mathcal{S}^{e}_1 + \mathcal{S}^{e}_2$ are of the form $x(n) + x'(n + m)$ and $x(n) + x'(n - m - 1)$ where $n$, $n + m$ and $n - m -1$ take even values. When $m$ is odd, note that $n + m$ is odd and $n - m - 1$ is even. Hence, the set $\mathcal{S}^{e}_1 + \mathcal{S}^{e}_2$ will have no points on $O^{m}$ and $M/2$ points on $I^{m}$. Similarly, when $m$ is even, $\mathcal{S}^{e}_1 + \mathcal{S}^{e}_2$ has no points on $I^{m}$ and $M/2$ points on $O^{m}$.\\
Therefore, $\mathcal{S}^{e}_1 + \mathcal{S}^{e}_2$ will have points on the following set of circles $\left\lbrace I^{M/2 - 1}, O^{M/2 - 2}, I^{M/2-3}, \cdots , I^{1}, O^{0} \right\rbrace$. Since $n$ takes only even values, using observation 1) of Proposition \ref{prop5}, $\phi_{min}$ between the points of $\mathcal{S}^{e}_1 + \mathcal{S}^{e}_2$ for a given $m$ is $\frac{4\pi}{M}$. Hence the points of $\mathcal{S}^{e}_1 + \mathcal{S}^{e}_2$ are maximally separated on every circle. From observation 4) of Proposition \ref{prop5}, $\phi_{min}$ between the points in $I^{q}$ and $O^{q-1}$ is $\frac{2\pi}{M}$ for all $q = 1$ to $M/2 - 1$. Similarly $\phi_{min}$ between the points in $O^{q}$ and $I^{q-1}$ is $0$ for all $q = 2$ to $M/2 - 2$. It can be verified that the structure of $\mathcal{S}^{o}_1 + \mathcal{S}^{o}_2$ is identical to that of $\mathcal{S}^{e}_1 + \mathcal{S}^{e}_2$.
\begin{proposition}
\label{prop8}
For $M \geq 8$, $r(I^{q})$ satisfies the following inequality for $q = 1$ to $M/2 - 3$,
\begin{equation*}
d(r(I^{q}), r(I^{q})e^{i\frac{2\pi}{M}}) \geq 2r(I^{M/2-1})\mbox{sin}(\frac{2\pi}{M}).
\end{equation*}
\end{proposition}
\begin{proof}
We know that $2r(I^{M/2-1})\mbox{sin}(\frac{2\pi}{M}) = 4\mbox{sin}(\frac{\theta}{2})\mbox{sin}(\frac{2\pi}{M})$ and $d(r(I^{q}), r(I^{q})e^{i\frac{2\pi}{M}}) = 4\mbox{cos}(\frac{\pi(q + 1)}{M} - \frac{\theta}{2})\mbox{sin}(\frac{\pi}{M})$. Since $q \leq M/2 - 3$, $\frac{\pi(q + 1)}{M} - \frac{\theta}{2} \leq \frac{\pi}{2} - \frac{2\pi}{M} - \frac{\theta}{2} \leq \frac{\pi}{2}$. Hence, for all values of $\theta$ and $1 \leq q \leq M/2 - 3$, $\mbox{cos}(\frac{\pi}{2} - \frac{2\pi}{M} - \frac{\theta}{2}) \geq \mbox{cos}(\frac{\pi}{2} - \frac{2\pi}{M}) = \mbox{sin}(\frac{2\pi}{M})$. Therefore,
\begin{equation*}
\frac{2r(I^{M/2-1})\mbox{sin}(\frac{2\pi}{M})}{d(r(I^{q}), r(I^{q})e^{i\frac{2\pi}{M}})} \leq \frac{\mbox{sin}(\frac{\theta}{2})}{\mbox{sin}(\frac{\pi}{M})} \leq 1.
\end{equation*} 
This completes the proof.
\end{proof}
\begin{proposition}
\label{prop9}
$r(O^{q-1})$ and $r(I^{q})$ satisfies the following inequality for $q = M/2 - 1$,
\begin{equation*}
d(r(O^{q-1}), r(I^{q})) \geq 2r(I^{M/2-1})\mbox{sin}(\frac{2\pi}{M}).
\end{equation*}
\end{proposition}
\begin{proof}
The inequality is straightforward to prove using standard trigonometric identities.
\end{proof}
\begin{lemma}
For $0 < \theta < \frac{\pi}{M}$, the minimum distance of the sets $\mathcal{S}^{ee}_{sum}$ and $\mathcal{S}^{oo}_{sum}$ is
\begin{equation}
\label{min_dist_set_2}
d^{ee}_{min} = d^{oo}_{min} =  4\mbox{sin}\left( \frac{\theta}{2}\right)\mbox{sin}\left(\frac{2\pi}{M}\right) .
\end{equation}
\end{lemma}
\begin{proof}
See Appendix \ref{proof_lemma_2}.
\end{proof}
\subsection{Optimality of Ungerboeck partitioning}
\label{sec3_subsec3}
In the preceding subsection, $d_{min}$ values of each one of the sets of $\mathcal{A}$ induced by Ungerboeck partition on $\mathcal{S}_{1}$ and $\mathcal{S}_{2}$ have been computed in \eqref{min_dist_set_1} and \eqref{min_dist_set_2}. Note that when $\theta = \frac{\pi}{M}$, $d^{ee}_{min} = d^{eo}_{min}$. In this subsection, using these values, we show that for $\theta = \frac{\pi}{M}$, a non-Ungerboeck partition on $\mathcal{S}_{1}$ and $\mathcal{S}_{2}$ results in a set $\mathcal{A}$ such that the $d_{min}$ of at least one of the sets in $\mathcal{A}$ is lesser than $d^{ee}_{min}$.
\begin{theorem}
\label{ungerboeck_theorem}
For $\theta = \frac{\pi}{M}$, Ungerboeck partitioning on $\mathcal{S}_{1}$ and $\mathcal{S}_{2}$ in to two sets is optimal in maximizing the minimum of the $d_{min}$ values of the sets in $\mathcal{A}$.
\end{theorem}
\begin{proof}
Let $\mathcal{S}^{1}_{i}$ and $\mathcal{S}^{2}_{i}$ be the two sets resulting from a partition of $\mathcal{S}_{i}$ for $i = 1, 2$. If either $\mathcal{S}_{1}$ or $\mathcal{S}_{2}$ is not Ungerboeck partitioned, then it is shown that, $d_{min}$ of at least one of the sets in the set $\mathcal{A} = \left\lbrace \mathcal{S}^{i}_{1} + \mathcal{S}^{j}_{2} ~ |~ \forall ~ i, j = 1, 2 \right\rbrace$ is lesser than $4\mbox{sin}(\frac{\pi}{2M})\mbox{sin}(\frac{2\pi}{M})$ (we have substituted $\theta = \frac{\pi}{M}$ in \eqref{min_dist_set_2}). In other words, $\phi_{min}$ between the points of at least one of sets in $\mathcal{A}$ is smaller than $\frac{4\pi}{M}$ on $O^{M/2 -1}$ (since $\theta = \frac{\pi}{M}$, note that $O^{m}$ = $I^{m}$ for all $m$). It is assumed that there are exactly $M/2$ points on $O^{M/2 - 1}$ in each set of $\mathcal{A}$. Otherwise, at least one set contains more than $M/2$ points and hence $\phi_{min}$ between the points in that set on $O^{M/2 - 1}$ can be at most $\frac{3\pi}{M}$. Therefore, the sub-optimality of the partition is proved.  Without loss of generality, assume that $x(a), x(a+ 1) \in \mathcal{S}^{1}_{1}$ for some $a$ such that $0 \leq a \leq M-2$. Irrespective of the partition of $\mathcal{S}_{2}$, either $x'(a+ \frac{M}{2}) \in \mathcal{S}^{1}_{2}$ or $x'(a+ \frac{M}{2}) \in \mathcal{S}^{2}_{2}$. We assume that $x'(a+ \frac{M}{2}) \in \mathcal{S}^{1}_{2}$. Hence, $x'(a+ \frac{M}{2})$ + $x(a)$, $x'(a+ \frac{M}{2})$ + $x(a + 1) \in \mathcal{S}^{1}_{1}$ + $\mathcal{S}^{1}_{2}$ which lie on $O^{M/2 - 1}$. Note that, $x'(a+ \frac{M}{2})$ + $x(a)$ and $x'(a+ \frac{M}{2})$ + $x(a + 1)$ have an angular separation of $\frac{(N + 1)\pi}{N}$. With this, the rest of the $M/2 - 2$ points get placed on $O^{M/2 - 1}$ in one of the following two ways : (i) at least $M/4 - 1$ points appear between $0$ and -$\frac{\pi(M - 1)}{M}$ radians due to which at most $M/4 - 1$ points lie between $0$ to $\frac{\pi(M + 1)}{M}$ radians (ii) at least $M/4$ points appear between $0$ and $\frac{\pi(M + 1)}{M}$ radians due to which at most $M/4 - 2$ points lie between $0$ and -$\frac{\pi(M - 1)}{M}$ radians. For the case in (i), at least $M/4$ points have to appear from $0$ and -$\frac{\pi(M - 1)}{M}$ radians, due to which $\phi_{min}$ between them can at most be $\frac{4\pi (M - 1)}{M^2}$ which is less than $\frac{4\pi}{M}$. Similarly, for the case in (ii), it can be shown that $\phi_{min}$ between the set of points which lie from $0$ to $\frac{\pi(M + 1)}{M}$ radians is lesser than $\frac{4\pi}{M}$. This completes the proof.
\end{proof}

\indent For $M$-PSK signal sets, when $\theta \neq \frac{\pi}{M}$, the optimal partitioning on $\mathcal{S}_{1}$ and $\mathcal{S}_{2}$ is not known. However, we present an example wherein for a particular value of $\theta$, a non-Ungerboeck partition on $\mathcal{S}_{1}$ and $\mathcal{S}_{2}$ results in a set $\mathcal{A}$ such that the $d_{min}$ of all the sets in $\mathcal{A}$ is larger than min ($d^{ee}_{min}$, $d^{eo}_{min}$).

\begin{example}
\label{counter_example_labelling}
$\mathcal{S}_{1}$ is a uniform 8-PSK signal set with $\theta = \frac{\pi}{25}$. The partition of $\mathcal{S}_{1}$ and $\mathcal{S}_{2}$ are given by,
\begin{equation*}
\mathcal{S}^{1}_{1} = \left\lbrace x(1), x(2), x(4), x(6) \right\rbrace,    
\end{equation*}
\begin{equation*}
\mathcal{S}^{1}_{2} = \left\lbrace x'(1), x'(4), x'(5), x'(8) \right\rbrace,
\end{equation*}
\begin{equation*}
\mathcal{S}^{2}_{1} = \left\lbrace x(3), x(5), x(7), x(8) \right\rbrace  \mbox{ and } 
\end{equation*}
\begin{equation*}
\mathcal{S}^{2}_{2} = \left\lbrace x'(2), x'(3), x'(6), x'(7) \right\rbrace.\\
\end{equation*}

\end{example}

\indent Nevertheless, in the following theorem, we show that, for some class of partitions, (for any $\theta \in \left(0, \frac{\pi}{M} \right)$), minimum Euclidean distance of at least one of the sets in $\mathcal{A}$ is lesser than min ($d^{ee}_{min}$, $d^{eo}_{min}$).
\begin{theorem}
For $\theta \in \left(0, \frac{\pi}{M} \right)$, if the partition of $\mathcal{S}_{1}$ and $\mathcal{S}_{2}$ are such that $x(a), x(a + 1)$ and $x(a +2) \in \mathcal{S}^{j}_{1}$ for some $j = 1, 2$ and $0 \leq a \leq M-3$, then the $d_{min}$ of at least one of the sets in $\mathcal{A}$ is lesser than min ($d^{ee}_{min}$, $d^{eo}_{min}$).
\end{theorem}
\begin{proof}
Assume $x(a),  x(a + 1)$ and $x(a +2) \in \mathcal{S}^{1}_{1}$. Among the three points, $x'(a - M/2), x'(a + 1- M/2)$ and $x'(a +2 - M/2),$ two of them must belong to either $\mathcal{S}^{1}_{2}$ or $\mathcal{S}^{2}_{2}$. Without loss of generality, assume two of them belong to $\mathcal{S}^{1}_{2}$. If $x'(a - M/2), x'(a + 1- M/2) \in \mathcal{S}^{1}_{2}$, then $x(a) + x'(a - M/2)$, $x(a + 1) + x'(a + 1 - M/2) \in \mathcal{S}^{1}_{1} + \mathcal{S}^{1}_{2}$ such that the two points lie on $I^{M/2 - 1}$ with an angular separation of $\frac{2\pi}{M}$. Hence the distance between the two points is lesser than $d^{ee}_{min}$. Same result can be proved if $x'(a + 1 - M/2), x'(a + 2 - M/2) \in \mathcal{S}^{1}_{2}$. Finally, if $x'(a - M/2), x'(a + 2 - M/2) \in \mathcal{S}^{1}_{2}$, then $x(a) + x'(a - M/2)$, $x(a + 2) + x'(a + 2 - M/2) \in \mathcal{S}^{1}_{1} + \mathcal{S}^{1}_{2}$ lies on $I^{M/2 - 1}$ with an angular separation of $\frac{4\pi}{M}$.
\end{proof}
\subsection{TCM with $M$-PAM signal sets}
\label{sec3_subsec4}
In the previous subsection, a systematic method of labelling the trellis pair $(T_{1}$, $T_{2})$ has been obtained when $M$-PSK signal sets (with $\theta^{*} = \frac{\pi}{M}$) are employed by both the users. In this subsection, we consider designing TCM schemes with $M$-PAM signal sets for both the users. For such a set-up, using the metric presented in Theorem \ref{thm}, it can be verified that $\theta^{*} = \frac{\pi}{2} ~ \forall M$ and for all values of SNR. Recall that, when $M$-PSK signal sets are employed, $\mathcal{S}_{sum}$ takes the structure of concentric PSK signal sets. However, when $M$-PAM signal sets are used, $\mathcal{S}_{sum}$ is a regular $M^{2}$-QAM (since $\forall \mbox{SNR}, \theta^{*} = \frac{\pi}{2}$). In this set-up, for a chosen trellis pair, the destination sees the corresponding sum trellis, $T_{sum}$ labelled with symbols from a $M^{2}$-QAM signal set. If the destination decodes for every $l$ channel uses and $\textbf{x}_{1} \in \mathcal{S}_{1}^{l}$ ($\mathcal{S}_{1}$ = $M$-PAM signal set) and $\textbf{x}_{2} \in \mathcal{S}_{2}^{l}$ $(\mathcal{S}_{2} = e^{\frac{i\pi}{2}}\mathcal{S}_{1}$) represent the codewords of User-1 and User-2 respectively, the received sequence at the destination is given by $\textbf{y} = \textbf{x}_{sum} + \textbf{n}$ where $\textbf{x}_{sum} = \textbf{x}_{1} + \textbf{x}_{2} \in \mathcal{S}_{sum}^{l}$ (where $\mathcal{S}_{sum} = M^{2}$-QAM) and $\textbf{n} \sim \mathcal{CN} \left(0, \sigma^{2}\textbf{I}_{l}\right)$. The decoding metric is given by,
\begin{equation*}
\hat{\textbf{x}}_{sum} = \mbox{arg} \min_{\textbf{x}_{sum}} ||\textbf{y} - \textbf{x}_{sum}||^{2}.
\end{equation*}
Since $\textbf{y}_{I}$ and $\textbf{y}_{Q}$ respectively are dependent on $\textbf{x}_{1}$ and $\textbf{x}_{2}$ only (where $\textbf{y}_{I}$ and $\textbf{y}_{Q}$ denote the in-phase and quadrature components of $\textbf{y}$), the above decoding metric splits as follows,
\begin{equation*}
\hat{\textbf{x}}_{1} = \mbox{arg} \min_{\textbf{x}_{1} \in \mathcal{C}_{1}} ||\textbf{y}_{I} - \textbf{x}_{1}||^{2} \mbox{ and } \hat{\textbf{x}}_{2} = \mbox{arg} \min_{\textbf{x}_{2} \in \mathcal{C}_{2}} ||i\textbf{y}_{Q} - \textbf{x}_{2}||^{2}.
\end{equation*}
Therefore, the destination can decode for a sequence over $M$-PAM alphabet on the individual trellises $T_{1}$ and $T_{2}$ instead of decoding for a sequence over $M^{2}$ QAM alphabet on $T_{sum}$. Since, decoding for the symbols of one user is independent of the decoding for the symbols of the other, trellises $T_{1}$ and $T_{2}$ has to be labelled based on Ungerboeck rules as done for a SISO-AWGN channel. Hence, all the TCM based trellis codes with $M$-PAM alphabets existing for SISO AWGN are applicable in the two-user GMAC setup. With this, the decoding complexity at the destination is significantly reduced as the state complexity profile of the trellis over which the decoder works is $\left\lbrace q_{i,0}, ~q_{i,1}, \cdots ~q_{i,n} \right\rbrace$ (when decoding for User-i) instead of $\left\lbrace q_{1,0}q_{2,0}, ~q_{1,1}q_{2,1}, \cdots ~q_{1,n}q_{2,n} \right\rbrace$. In general, when a complex signal set is used by either one of the users, the destination has to necessarily decode for a sequence over $\mathcal{S}_{sum}$ on $T_{sum}$ which has high decoding complexity.\\
\indent From the above discussion, it is clear that for a two-user GMAC, one dimensional signal sets can be preferred over complex signal sets for reducing the decoding complexity. However, it is not clear if there is any loss in the CC sum capacity by using single dimensional signal sets. As a first step towards answering the above question, in Fig. \ref{sum_alphabet_capacity_qam_pam}, we have plotted the sum CC capacity (i.e. $R_{1} + R_{2}$) as a function of SNR for two scenarios; (i) when QPSK signal sets are used with angles of rotation as given in Table \ref{rotation_table1} and (ii) when 4-PAM signal sets are used with $\theta^{*} = \frac{\pi}{2}$. For both the scenarios, average energy per symbol per user is made the same. As shown in the plot, there is a marginal difference in the CC sum capacity between the two schemes and in particular, at high SNR the sum capacity of the later scheme is larger than the former. Therefore, using $4$-PAM signal sets provide reduced decoding complexity with \textit{almost} the same CC sum capacity as that of QPSK signal sets. Similar curves have been obtained in Fig. \ref{sum_alphabet_capacity_qam_bpsk} for the following two scenarios (i) when User-1 and User-2 uses QPSK and BPSK signal set respectively (with appropriate angle of rotation) and (ii) when User-1 uses 4-PAM signal set, User-2 uses BPSK with $\theta^{*} = \frac{\pi}{2}$.
\begin{figure}
\centering
\includegraphics[width=3.5in]{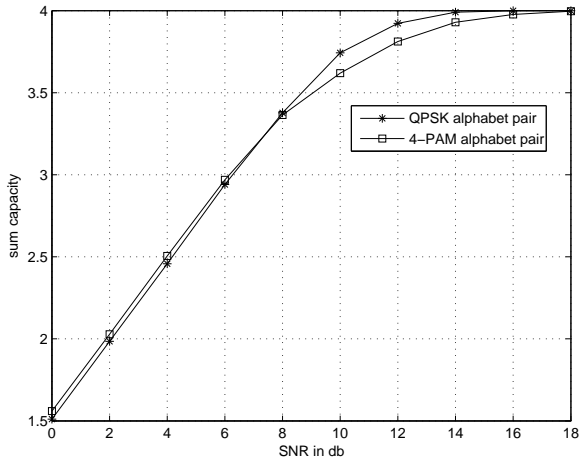}
\caption{Capacity of the sum alphabet of QPSK signal sets and 4-PAM signal sets with optimal rotation.} 
\label{sum_alphabet_capacity_qam_pam}
\end{figure}
\begin{figure}
\centering
\includegraphics[width=3.5in]{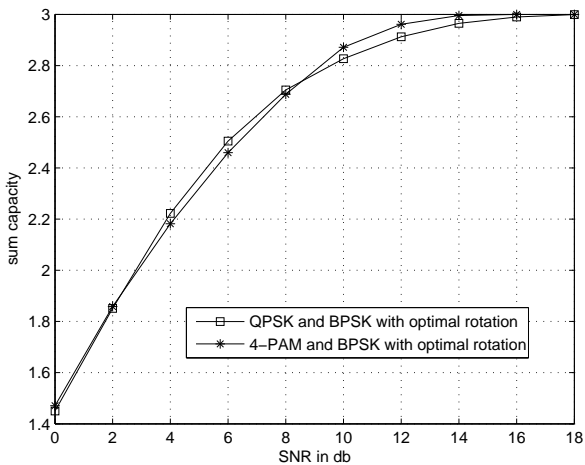}
\caption{Capacity of the sum alphabet of QPSK/BSK signal sets and 4-PAM/BPSK signal sets with optimal rotation.} 
\label{sum_alphabet_capacity_qam_bpsk}
\end{figure}

\indent For arbitrary values of $M$, we conjecture that $M$-PAM signal sets (with a relative rotation of $\frac{\pi}{2}$) provide sum capacities which are marginally close to that of $M$-PSK and $M$-QAM signal set pairs (with appropriate rotation) with the same average energy. Note that the above relation can also be observed in a two-user GMAC with Gaussian code alphabets. If $x_{1}, x_{2} \sim \mathcal{CN} \left(0, \frac{\rho}{2}\right)$, the received symbol at the destination is given by $y = x_{1} + x_{2} + n$ where we assume that $n \sim \mathcal{N} \left(0, \frac{1}{2}\right) $ in each dimension. The sum capacity for the above model is
\begin{equation*}
\mbox{log}_{2}(1 + \frac{\rho}{2}) + \mbox{log}_{2}(1 + \frac{\rho}{2 + \rho}) = \mbox{log}_{2}(1 + \rho).
\end{equation*}
Note that, the capacities for User-1 and User-2 are $\mbox{log}_{2}(1 + \frac{\rho}{2})$ and $\mbox{log}_{2}(1 + \frac{\rho}{2 + \rho})$ respectively.  However, if $x_{1}\sim \mathcal{N} \left(0, \frac{\rho}{2}\right)$ and $x_{2} = ix_{2}'$ such that $x_{2}' \sim \mathcal{CN} \left(0, \frac{\rho}{2}\right)$, the capacity in each dimension is
\begin{equation*} 
\frac{1}{2}\mbox{log}_{2}(1 + \rho)
\end{equation*}
and hence the sum capacity is $\mbox{log}_{2}(1 + \rho)$ which is equal to the sum capacity of complex Gaussian alphabets. For the later scheme, the capacity for User-1 and User-2 is $\frac{1}{2}\mbox{log}_{2}(1 + \rho)$. When the individual capacities for each user are compared between the two schemes, it is clear that, for one of the users, the capacity will be larger in the later scheme and smaller in the former scheme whereas for the other user, capacity will be larger in the former scheme and smaller in the later scheme there by making the sum capacity of both the schemes equal. Hence, the capacity region of real Gaussian alphabets (with $\theta = \frac{\pi}{2}$) lies inside the capacity region of complex Gaussian alphabets with only one point of intersection.\\
\indent In a SISO AWGN channel, it is well known that, for a given SNR, one dimensional signal sets incur some loss in the CC capacity when compared to \textit{well packed} complex signal sets having the same average energy and equal number of points. Note that, the CC capacity of individual signal sets, $\mathcal{S}_{1}$ and $\mathcal{S}_{2}$ are of little importance in the GMAC set-up, since for an input alphabet pair $(\mathcal{S}_{1}, \mathcal{S}_{2})$, the destination sees an equivalent AWGN channel with the corresponding $\mathcal{S}_{sum}$ as its input (neither $\mathcal{S}_{1}$ nor $\mathcal{S}_{2}$).  Hence, in order to maximize the sum capacity, the alphabet pair $(\mathcal{S}_{1}, \mathcal{S}_{2})$ has to be chosen such that CC capacity of $\mathcal{S}_{sum}$ is maximized. Since we have shown that, for a given SNR, the sum capacity of $4$-PAM alphabet pair is marginally close to that of a QPSK alphabet pair, we conjecture that for any $M$, $M$-PAM alphabet pairs (with $\theta^{*} = \frac{\pi}{2}$) do not incur significant loss in the sum capacity when compared to $M$-PSK and $M$-QAM alphabet pairs in a two-user GMAC.
\subsubsection{Examples and Numerical results}
\begin{figure}
\centering
\includegraphics[width=2in]{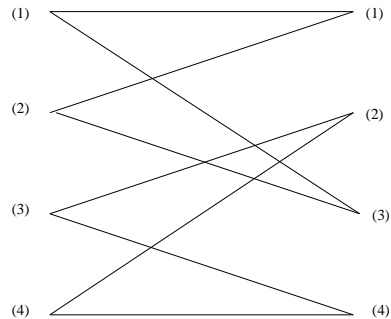}
\caption{Trellis structure employed by both the users.} 
\label{four_state_trellis}
\end{figure}
\begin{figure}
\centering
\includegraphics[width=2.5in]{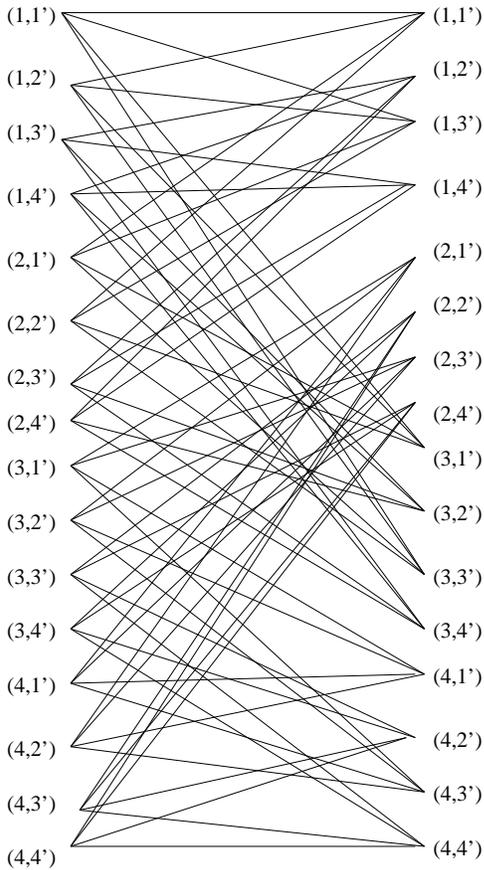}
\caption{$T_{sum}$ for the trellis pair presented in Fig. \ref{four_state_trellis}.} 
\label{four_state_sum_trellis}
\end{figure}
In this subsection, we present numerical results on the minimum accumulated squared Euclidean distance, $d^{2}_{free, min}$ of $T_{sum}$ when the trellis presented in Fig. \ref{four_state_trellis} is employed by both the users using QPSK and $4$-PAM signal sets. For the trellis in Fig. \ref{four_state_trellis}, $T_{sum}$ is as shown in Fig. \ref{four_state_sum_trellis}. We compute $d^{2}_{free, min}$ for the following two scenarios (i) when the individual trellises are labelled with unit energy QPSK signal sets (with an angle of rotation $\frac{\pi}{4}$) using the Ungerboeck rules. (ii) when the trellises are labelled with $4$-PAM signal sets, $\sqrt(\frac{1}{5})\left\lbrace -3, -1, 1, 3 \right\rbrace$ (with $\theta^{*} = 90^{o}$) using Ungerboeck rules. For scenario (i), $d^{2}_{free, min}$ is 5.8578 where as for scenario (ii), $d^{2}_{free, min}$ is 7.20. Hence, the asymptomatic coding gain of $0.89$ db can be obtained by using $4$-PAM signal sets over QPSK signal sets. When the alphabets in the scenarios discussed are used on the trellis pair presented in Fig. \ref{trp}, the corresponding asymptotic coding gain is $0.57$ db.
\section{Space time coding for two-user MIMO-MAC}
\label{sec4}
In this section, we introduce a two-user MIMO (Multiple Input Multiple Output) MAC model and propose construction of two different classes STBC pairs with certain nice properties. In the following subsection, the MIMO-MAC model considered in this paper has been described.
\subsection{Channel model of two-user MIMO-MAC}
\label{sec4_subsec1}
The two-user MIMO-MAC model considered in this paper consists of two sources each equipped with $N_{t}$ antennas and a destination equipped with a single antenna. The channel from the $i$-th antenna of the $j$-th user to the destination is a quasi-static block fading channel denoted by $h_{ji}$  $\forall ~i = 1$ to $N_{t}$ and $j = 1, 2$ where each $h_{ji} \sim \mathcal{CN} \left(0, 1\right)$ with the coherence time interval of at least $l$ channel uses. For each $j$, if $\textbf{x}_{j} \in \mathbb{C}^{1 \times N_{t}}$ is the vector transmitted by User-$j$ such that every symbol of $\textbf{x}_{j}$ has average unit energy, the received symbol at the destination for every channel use is given by  
\begin{equation}
\label{mimo_mac_channel}
y = \sqrt{\frac{\rho}{2N_{t}}}\textbf{x}_{1}\textbf{h}_{1} + \sqrt{\frac{\rho}{2N_{t}}}\textbf{x}_{2}\textbf{h}_{2} + n,
\end{equation}
where $n \sim \mathcal{CN} \left(0, 1\right)$ is the additive noise at the destination, $\textbf{h}_{j}^{T} = \left[h_{j1} ~ h_{j2} \cdots h_{jN_{t}} \right]$ and $\rho$ is the average receive SNR at the destination. Throughout the paper, we assume the perfect knowledge of CSI (Channel State Information) at the destination which is commonly referred as CSIR. The two-user MIMO-MAC model described above is referred as a $(N_{t}, N_{t}, 1)$ MIMO-MAC. It is clear that the sum capacity of a $(N_{t}, N_{t}, 1)$ MIMO-MAC is equal to the capacity of a $2N_{t} \times 1$ MIMO channel (with CSIR) which is given by
\begin{equation}
\label{sum_capacity_fading _mac_channel}
C(N_{t}, N_{t}, 1) = E\left[ \mbox{log}_{2}\left( 1 + \frac{\rho}{2N_{t}}\left( \textbf{h}_{1}\textbf{h}_{1}^{H} + \textbf{h}_{2}\textbf{h}_{2}^{H}\right) \right) \right] 
\end{equation}
where the expectation is over the random variables $|h_{ji}|^{2}~ \forall ~i,j$. We also assume the perfect knowledge of the phase component of $h_{ji}$ at the $j$-th user $\forall i, j$ which we refer as CSIT-P. The $(N_{t}, N_{t}, 1)$ MIMO-MAC with the assumption of CSIT-P is referred as the $(N_{t}, N_{t}, 1, \textit{p})$ MIMO-MAC where $\textit{p}$ highlights the assumption of CSIT-P in the channel model. Note that, we do not assume the complete knowledge of $h_{ji}$ at the transmitters in which case, optimal power allocation techniques can be applied to improve the system performance. Since CSIT-P is known, each transmit antenna can compensate for the rotation introduced by the channel and hence the channel equation \eqref{mimo_mac_channel} can be written as
\begin{equation}
\label{mimo_mac_channel_csitp}
y = \sqrt{\frac{\rho}{2N_{t}}}\textbf{x}_{1}\tilde{\textbf{h}}_{1} + \sqrt{\frac{\rho}{2N_{t}}}\textbf{x}_{2}\tilde{\textbf{h}}_{2} + n,
\end{equation} 
where $\tilde{\textbf{h}}_{j}^{T} = \left[|h_{j1}| ~ |h_{j2}| \cdots |h_{jN_{t}}| \right]$. Suppose, $C(N_{t}, N_{t}, 1, \textit{p})$ denotes the sum capacity of a $(N_{t}, N_{t}, 1, \textit{p})$ MIMO-MAC with CSIR, it is straightforward to verify that $C(N_{t}, N_{t}, 1, \textit{p})$ = $C(N_{t}, N_{t}, 1)$. In the rest of this paper, a $(N_{t}, N_{t},1, \textit{p})$ MIMO-MAC is denoted as a $N_{t}$-MIMO-MAC.\\
\indent The sum capacity of a $N_{t}$-MIMO-MAC (which is given by $C(N_{t}, N_{t}, 1, p)$) is computed by assuming that independent vectors are transmitted every time instant from both the users. However, when a Space Time Block Code (STBC) pair $(\mathcal{C}_{1}, \mathcal{C}_{2})$ is employed, the vectors transmitted at every time instant will not be independent. Let the dimensions of the STBC used by both the users be $l \times N_{t}$ (where $l$ denotes the number of complex channel uses). Throughout the paper, we assume that STBCs for both the users have the same dimensions. If the $l \times N_{t}$ matrices transmitted by User-1 and User-2 are $\textbf{X}$ and $\textbf{Y}$ respectively, then the received vector, $\textbf{y} \in \mathbb{C}^{l}$ is given by
\begin{equation}
\label{MIMO_chanel_with_STBC}
\textbf{y} = \sqrt{\frac{\rho}{2N_{t}}}\textbf{X}\tilde{\textbf{h}}_{1} + \sqrt{\frac{\rho}{2N_{t}}}\textbf{Y}\tilde{\textbf{h}}_{2} + \textbf{n},
\end{equation}
where $\textbf{n}$ denoted the complex $l \times 1$ additive noise vector. If the STBCs used are of rate $R$ complex symbols per channel use, then there are $lR$ independent complex variables for each user describing the corresponding matrix. Let the vector containing $lR$ variables of $\textbf{X}$ and $\textbf{Y}$ be denoted by $\textbf{x} \in \mathbb{C}^{lR \times 1}$ and $\textbf{y} \in \mathbb{C}^{lR \times 1}$ respectively. Totally, there are $2lR$ independent variables denoted by $\textbf{z} \in \mathbb{C}^{2lR \times 1}$ where $\textbf{z} = \left[ \textbf{x}^{T} ~ \textbf{y}^{T}\right]^{T}$. If $\textbf{X}$ and $\textbf{Y}$ are from linear designs, we can write \eqref{MIMO_chanel_with_STBC} as given below \cite{HaH}
\begin{equation}
\textbf{y} = \sqrt{\frac{\rho}{2N_{t}}}\tilde{\textbf{H}}\textbf{z} +  \textbf{n},
\end{equation} 
where $\tilde{\textbf{H}} \in \mathbb{C}^{l \times lR}$. The capacity of this new channel, $\tilde{\textbf{H}}$ is the capacity of a collocated MIMO channel with $lR$ transmit antennas and $l$ receive antennas given by
\begin{equation*}
E~\left[\mbox{log}_{2}\left(\mbox{det}\left( \textbf{I}_{l} + \frac{\rho}{2N_{t}}\hat{\textbf{H}}\hat{\textbf{H}}^{H}\right)\right)\right].  
\end{equation*}
Therefore, after introducing the STBC pair $\left(\mathcal{C}_{1}, \mathcal{C}_{2}\right)$, the maximum mutual information between the vector $\textbf{z}$ and $\textbf{y}$, $I(\textbf{z} : \textbf{y}~|~\tilde{\textbf{H}})$ is given by
\begin{equation*}
C_{\mbox{\begin{small}STBC\end{small}}}(N_{t}, N_{t}, 1, \textit{p}) = \frac{1}{l}E~\left[\mbox{log}_{2}\left(\mbox{det}\left( \textbf{I}_{l} + \frac{\rho}{2N_{t}}\hat{\textbf{H}}\hat{\textbf{H}}^{H}\right)\right)\right]  
\end{equation*}
where the factor $\frac{1}{l}$ takes care of the rate loss due to coding across time. It is clear that the above value cannot be more than $C(N_{t}, N_{t}, 1, p)$. On the similar lines of the definition of information lossless STBCs for collocated MIMO channels \cite{SRS}, information lossless STBC pairs are defined below for a $N_{t}$-MIMO-MAC.
\begin{definition}
If the maximum mutual information, $I(\textbf{z} : \textbf{y}~|~\tilde{\textbf{H}})$ when an STBC pair $(\mathcal{C}_{1}, \mathcal{C}_{2})$ is used for a $N_{t}$-MIMO-MAC, is equal to the capacity of a $2N_{t} \times 1$ MIMO channel, then the pair $(\mathcal{C}_{1}, \mathcal{C}_{2})$ is called an information lossless STBC pair.
\end{definition}

\indent In the rest of the section, we propose two classes of STBC pairs from Real Orthogonal Designs (RODs) for a $N_{t}$-MIMO-MAC. For deriving certain properties of the codes that we are going to propose, the following definition and theorem are important.
\begin{definition}
Let the channel equation of a MISO (Multiple Input Single Output) system with $N_{t}$ transmit antennas be represented by $y = \textbf{x}\textbf{h} + n$ where $y$ is the received symbol at the destination, $n$ is the additive noise, $\textbf{h}$ is the $N_{t}$ length channel vector and $\textbf{x}$ is the input vector to the channel of length $N_{t}$. Such a MISO channel is referred as a single-dimensional MISO channel whenever $\textbf{x}, \textbf{h} \in \mathbb{R}^{N_{t}}$.
\end{definition}
\begin{theorem}
\label{rod_IL}
STBCs from the rate-1 ROD (which also includes rate-1 rectangular ROD) for $N_{t}$ antennas are information lossless for a single-dimensional $N_{t} \times 1$ MIMO channel for all values of $N_{t}$ .
\end{theorem}
\begin{proof}
See Appendix \ref{proof_theorem_rod_IL}.
\end{proof}

\indent Throughout the section, we assume that the destination performs joint decoding of the symbols of User-1 and User-2 by decoding for a $l \times 2N_{t}$ space-time codeword, $\textbf{Z} = \left[ \textbf{X} ~\textbf{Y}\right]$ in a virtual $2N_{t} \times 1$ MIMO channel (where $\left[ \textbf{X} ~\textbf{Y}\right]$ denotes juxtaposing of the matrices $\textbf{X}$ and $\textbf{Y}$). Therefore, applying the full diversity design criterion derived for space-time codes in point to point coherent MIMO channels \cite{TJC} on the set of codewords of the form $\textbf{Z}$, the diversity order of the code pair $(\mathcal{C}_{1}, \mathcal{C}_{2})$ in a $N_{t}$-MIMO-MAC is $N_{t}$ provided each space-time block code $\mathcal{C}_{i}$ is individually fully diverse for a point to point coherent MIMO channel.
\subsection{STBC pairs from Separable Orthogonal Designs (SODs) for a $N_{t}$-MIMO-MAC}
\label{sec4_subsec2}
In this section, we construct  STBC pairs $(\mathcal{C}_{1}, \mathcal{C}_{2})$ for a $N_{t}$-MIMO-MAC such that the ML-decoding complexity at the destination is reduced (where $\mathcal{C}_{1}$ is used by User-1 and $\mathcal{C}_{2}$ is used by User-2). The STBC pair $(\mathcal{C}_{1}, \mathcal{C}_{2})$ is specified by presenting a complex design pair $(\textbf{X}, \textbf{Y})$ and a signal set pair $(\mathcal{S}_{1}, \mathcal{S}_{2})$ such that $\mathcal{C}_{1}$ and $\mathcal{C}_{2}$ are generated by making the variables of $\textbf{X}$ and $\textbf{Y}$ take values from $\mathcal{S}_{1}$ and $\mathcal{S}_{2}$ respectively. In particular, we construct complex design pairs $(\textbf{X}, \textbf{Y})$ using the well known class of RODs. The proposed class of complex designs are introduced in the following definition.
\begin{definition}
\label{def_sod}
Let the $l \times N_{t}$ matrix $\textbf{X}$ represent a ROD in $k$ real variables for $N_{t}$ antennas. If every real variable of $\textbf{X}$ is viewed as a complex variable, then $\textbf{X}$ becomes a design in $k$ complex variables which we refer as a Separable Orthogonal Design (SOD). 
\end{definition}

\indent If a design $\textbf{X}$ represents a SOD, then from Definition \ref{def_sod}, it is clear that $\textbf{X}_{I}$ and  $\textbf{X}_{Q}$ are identical RODs. Also, since rate-1 RODs exist for $\forall N_{t}$, rate-1 SODs (in complex symbols per channel use) also exist for $\forall N_{t}$ \cite{Xl} (except for $N_{t}$ = 2, 4 and 8, note that all other SODs are rectangular designs). Throughout the paper, we only consider the class of rate-1 SODs. In the following example, we present a SOD pair for $4$-MIMO-MAC in $4$ complex variables per user.

\begin{example}
\label{example_saod}
For $N_{t} = 4, k = 4,$
\begin{equation*}
\textbf{X} =\left[\begin{array}{rrrr}
x_{1} & x_{2} & x_{3} & x_{4}\\
-x_{2} & x_{1} & -x_{4} & x_{3}\\
-x_{3} & x_{4} & x_{1} & -x_{2}\\
-x_{4} & -x_{3} & x_{2} & -x_{1}\\
\end{array}\right] \mbox{ and } 
\end{equation*}
\begin{equation*}
\textbf{Y} =\left[\begin{array}{rrrr}
y_{1} & y_{2} & y_{3} & y_{4}\\
-y_{2} & y_{1} & -y_{4} & y_{3}\\
-y_{3} & y_{4} & y_{1} & -y_{2}\\
-y_{4} & -y_{3} & y_{2} & -y_{1}\\
\end{array}\right].
\end{equation*}\\
\end{example}

\indent Towards generating the STBC pair $(\mathcal{C}_{1}, \mathcal{C}_{2})$, we restrict the complex variables of a SOD to take values from the class of regular-QAM signal sets only. The variables are precluded to take values from signal sets where the in-phase and quadrature components are entangled, for example, $M$-PSK signal sets. The advantage of choosing a regular-QAM signal set for $\mathcal{S}_{1}$ and $\mathcal{S}_{2}$ is described in the next subsection. When the SOD pair $(\textbf{X}, \textbf{Y})$ is used, the received vector at the destination is of the form
\begin{equation*}
\textbf{y} = \sqrt{\frac{\rho}{2N_{t}}}\textbf{X}\tilde{\textbf{h}}_{1} + \sqrt{\frac{\rho}{2N_{t}}}\textbf{Y}\tilde{\textbf{h}}_{2} + \textbf{n}.
\end{equation*}
Since the variables of the two designs take values from regular QAM signal sets and the channels $\tilde{\textbf{h}}_{j}$'s are real, the $N_{t}$-MIMO-MAC with STBC pairs from SOD pair $(\textbf{X}, \textbf{Y})$ splits in to two parallel single-dimensional $N_{t}$-MIMO-MACs with STBC pairs from ROD pairs $(\textbf{X}_{I}, \textbf{Y}_{I})$ and $(\textbf{X}_{Q}, \textbf{Y}_{Q})$ respectively. For each $\begin{small}\heartsuit\end{small} = I, Q$, the channel equation is given by
\begin{equation*}
\textbf{y}_{\heartsuit} = \sqrt{\frac{\rho}{2N_{t}}}\textbf{X}_{\heartsuit}\tilde{\textbf{h}}_{1} + \sqrt{\frac{\rho}{2N_{t}}}\textbf{Y}_{\heartsuit}\tilde{\textbf{h}}_{2} + \textbf{n}_{\heartsuit},
\end{equation*}
where $\textbf{n}_{\heartsuit} \sim \mathcal{N} \left(0, \frac{1}{2}\textbf{I}_{T}\right)$. Henceforth, we consider only one of the single-dimensional channels for all the analysis purposes. The following theorem shows that STBC pairs from SODs are information lossless for a $N_{t}$-MIMO-MAC $\forall N_{t}$.
\begin{theorem}
For a $N_{t}$-MIMO-MAC, STBC pairs from the rate-1 SOD pair are information lossless.
\end{theorem}
\begin{proof}
Let $\textbf{X}$ and $\textbf{Y}$ represent two $l \times N_{t}$ rate-1 SODs for $N_{t}$ antennas in the variables $x_{1}, x_{2} \cdots x_{l}$ and $y_{1}, y_{2} \cdots y_{l}$ respectively.  Using the above design pair, the channel equation along the in-phase component is
\begin{equation*}
\textbf{y}_{I} = \sqrt{\frac{\rho}{2N_{t}}}\textbf{X}_{I}\tilde{\textbf{h}}_{1} + \sqrt{\frac{\rho}{2N_{t}}}\textbf{Y}_{I}\tilde{\textbf{h}}_{2} + \textbf{n}_{I},
\end{equation*}
where $\textbf{X}_{I}$ and $\textbf{Y}_{I}$ are identical RODs in the variables $x_{1I}, x_{2I} \cdots x_{lI}$ and $y_{1I}, y_{2I} \cdots y_{lI}$ respectively. Note that $\textbf{X}$ and $\textbf{Y}$ have the following column vector representations,

{\small
\begin{equation*}
\textbf{X}_{I} = \left[\textbf{A}_{1}\textbf{x}_{I} ~\textbf{A}_{2}\textbf{x}_{I} ~\cdots ~\textbf{A}_{N_{t}}\textbf{x}_{I}\right]; ~\textbf{Y}_{I} = \left[\textbf{B}_{1}\textbf{y}_{I} ~\textbf{B}_{2}\textbf{y}_{I} ~\cdots ~\textbf{B}_{N_{t}}\textbf{y}_{I}\right]
\end{equation*}
}

\noindent where $\left\lbrace \textbf{A}_{i} ~|~ i = 1 \mbox{ to } N_{t}\right\rbrace$ and $\left\lbrace \textbf{B}_{i} ~|~ i = 1 \mbox{ to } N_{t}\right\rbrace$ are the sets of column vector representation matrices of $\textbf{X}$ and $\textbf{Y}$ respectively and $\textbf{x}_{I}^{T} = \left[x_{1I} ~x_{2I} ~\cdots ~x_{lI} \right]$, $\textbf{y}_{I}^{T} = \left[y_{1I} ~y_{2I} ~\cdots ~y_{lI} \right]$. The channel equation along the in-phase component can also be written as,
\begin{equation}
\label{eq_I_channel}
\textbf{y}_{I} = \sqrt{\frac{\rho}{2N_{t}}}\hat{\textbf{H}}\textbf{z}_{I} + \textbf{n}_{I} \mbox{ where }
\end{equation}
where the $l \times 2N_{t}$ matrix, $\hat{\textbf{H}} = \left[\hat{\textbf{H}_{1}} ~ \hat{\textbf{H}_{2}}\right]$ such that $\hat{\textbf{H}_{1}} =  \sum_{i = 1}^{N_t} |h_{1i}|\textbf{A}_{i}$, $\hat{\textbf{H}_{2}} =  \sum_{i = 1}^{N_t} |h_{2i}|\textbf{B}_{i}$ and $\textbf{z}_{I} = \left[ \textbf{x}_{I}^{T} ~ \textbf{y}_{I}^{T}\right]^{T}$. The capacity of the channel in \eqref{eq_I_channel} is
\begin{equation*}
\frac{1}{2}E~\left[\mbox{log}_{2}\left(\mbox{det}\left( \textbf{I}_{l} + \frac{\rho}{2N_{t}}\hat{\textbf{H}}\hat{\textbf{H}}^{H}\right)\right)\right].
\end{equation*}
Since $\textbf{A}_{i}$'s and $\textbf{B}_{i}$'s are unitary and $\textbf{A}_{i}\textbf{A}_{j}^{T} + \textbf{A}_{j}\textbf{A}_{i}^{T} = \textbf{0}_{T \times T}$, $\textbf{B}_{i}\textbf{B}_{j}^{T} + \textbf{B}_{j}\textbf{B}_{i}^{T} = \textbf{0}_{T \times T}$ $\forall i, j$ such that $i \neq j$, we have $\hat{\textbf{H}}\hat{\textbf{H}}^{H} = \left(\textbf{h}_{1}\textbf{h}_{1}^{H} + \textbf{h}_{2}\textbf{h}_{2}^{H} \right)\textbf{I}_{l}$ and hence the capacity of a single-dimensional $N_{t}$-MIMO-MAC along the in-phase component with the SOD pair, ($\textbf{X}, \textbf{Y}$) is
\begin{equation*}
\frac{1}{2}E~\left[\mbox{log}_{2}\left(1 + \frac{\rho}{2N_{t}}\left(\textbf{h}_{1}\textbf{h}_{1}^{H} + \textbf{h}_{2}\textbf{h}_{2}^{H} \right)\right)\right].  
\end{equation*}
Similarly, the capacity of a single-dimensional $N_{t}$-MIMO-MAC along the quadrature component with the SOD pair, ($\textbf{X}, \textbf{Y}$) is 
\begin{equation*}
\frac{1}{2}E~\left[\mbox{log}_{2}\left(1 + \frac{\rho}{2N_{t}}\left(\textbf{h}_{1}\textbf{h}_{1}^{H} + \textbf{h}_{2}\textbf{h}_{2}^{H} \right)\right)\right].  
\end{equation*}
Therefore, the sum capacity is
\begin{equation*}
E~\left[\mbox{log}_{2}\left(1 + \frac{\rho}{2N_{t}}\left(\textbf{h}_{1}\textbf{h}_{1}^{H} + \textbf{h}_{2}\textbf{h}_{2}^{H} \right)\right)\right]
\end{equation*}
which is equal to $C(N_{t}, N_{t}, 1, p)$. Hence the SOD pair, $(\textbf{X}, \textbf{Y})$ is information lossless for a $N_{t}$-MIMO-MAC.
\end{proof}

\indent In the following subsection, we discuss the low ML decoding property of SODs.
\subsubsection{Low ML decoding complexity of SODs}
\label{sec4_subsec2_subsubsec2}
In this subsection, we show that STBC pairs from SODs are two-group decodable in a $N_{t}$-MIMO-MAC (in particular, we consider rate-1 SODs). For more details on STBCs with multi-group decodability for a collocated MIMO channel, we refer the reader to \cite{KaR2}. Since the designs $\textbf{X}$ and $\textbf{Y}$ are constructed using rate-1 RODs (wherein the number of real variables is equal to the number of channel uses), the destination has to decode a total of $4l$ real variables ($2l$ for each user) for every codeword use. Since a $N_{t}$-MIMO-MAC with STBC pairs from the SOD pair $(\textbf{X}, \textbf{Y})$ breaks down in to two parallel single-dimensional $N_{t}$-MIMO-MACs with STBCs from ROD pairs $(\textbf{X}_{I}, \textbf{Y}_{I})$ and $(\textbf{X}_{Q}, \textbf{Y}_{Q})$ respectively, for each $\heartsuit = I, Q$ the ML-decoding metric is given by
\begin{equation}
\hat{\textbf{X}}_{\heartsuit}, \hat{\textbf{Y}}_{\heartsuit} = \mbox{arg} \min_{\mathcal{C}_{1\heartsuit}, \mathcal{C}_{2\heartsuit}}||\textbf{y} - \sqrt{\frac{\alpha}{2N_{t}}}\textbf{X}_{\heartsuit}\tilde{\textbf{h}}_{1} + \sqrt{\frac{\alpha}{2N_{t}}}\textbf{Y}_{\heartsuit}\tilde{\textbf{h}}_{2}||^{2}.
\end{equation}
Therefore, along each dimension, the destination has to jointly decode only $2l$ real variables ($l$ variables of each user) for every codeword use which constitutes $l$ channel uses. For this set-up, the destination can use a sphere decoder in $\mathbb{R}^{l}$ to decode $l$ of the $2l$ real variables where as the remaining $l$ variables can be decoded using brute force search. Note that when either (i) CSIT-P is not available or (ii) if the users employ signal sets wherein the in-phase and the quadrature components of the complex variables are entangled, the destination has to jointly decode for $4l$ real variables ($2l$ variables of each user) in $\mathbb{R}^{2l}$ and hence the decoding complexity is increased. Note that since CSIT-P is available, the complex signal set used by one of the users can be relatively rotated with respect to the other to improve the performance. However, such rotations will only entangle the in-phase and quadrature components of the symbols there by increasing the decoding complexity as mentioned above. 
\subsection{STBC pairs from Real Orthogonal Designs for a $N_{t}$-MIMO-MAC}
\label{sec4_subsec3}
When STBC pairs from SODs are employed for a $N_{t}$-MIMO-MAC, it is clear that the signal transmitted by User-1 is an interference for User-2 and vice-verse. In this subsection, we propose a new class of STBC pairs from RODs wherein each user is interference free from the other. In the proposed scheme, User-1 employs a rate-$1$ ROD, $\textbf{X}$ for $N_{t}$ antennas and User-2 employs an identical ROD, $\textbf{Y}$. The variables of $\textbf{X}$ take values from a $M$-PAM signal set where as the variables of $\textbf{Y}$ take values from a signal set which is 90 degrees rotated version of signal set used for $\textbf{X}$. In general, both users can use PAM signal sets with different number of points. Since rate-$1$ RODs exist $\forall N_{t}$, the proposed scheme is also applicable for a $N_{t}$-MIMO-MAC $\forall N_{t}$.
\begin{example}
For a $4$-MIMO-MAC, the designs, $\textbf{X}$ and $\textbf{Y}$ are as given in Example \ref{example_saod} where
the variables $x_{1}, x_{2} \cdots x_{4}$ can take values from $\mathcal{S}_{1} = \left\lbrace -3, -1, 1, 3 \right\rbrace$ and $y_{1}, y_{2} \cdots y_{4}$ can take values from $\mathcal{S}_{2} = \left\lbrace -3i, -1i, 1i, 3i \right\rbrace$.
\end{example}

\indent In the proposed scheme, the received vector at the destination is of the form
\begin{equation*}
\textbf{y} = \sqrt{\frac{\rho}{2N_{t}}}\textbf{X}\tilde{\textbf{h}}_{1} + \sqrt{\frac{\rho}{2N_{t}}}\textbf{Y}\tilde{\textbf{h}}_{2} + \textbf{n}.
\end{equation*}
Since $\tilde{\textbf{h}}_{j}$'s are real vectors and the two designs take values from orthogonal signal sets, it is clear that the two users are interference free from each other. With this the $N_{t}$-MIMO-MAC splits in to two parallel MISO channels (one for each user) such that the MISO channel from (i) User-1 to the destination and (ii) User-2 to the destination are given in \eqref{user_1_channel} and \eqref{user_2_channel} respectively.
\begin{equation}
\label{user_1_channel}
\textbf{y}_{I} = \sqrt{\frac{\rho}{2N_{t}}}\textbf{X}\tilde{\textbf{h}}_{1} + \textbf{n}_{I}.
\end{equation}
\begin{equation}
\label{user_2_channel}
i\textbf{y}_{Q} = \sqrt{\frac{\rho}{2N_{t}}}\textbf{Y}\tilde{\textbf{h}}_{2} + i\textbf{n}_{Q}.
\end{equation}
\subsubsection{Capacity of a $N_{t}$-MIMO-MAC with RODs}
\label{sec4_subsec3_subsubsec1}
Note that the channels in \eqref{user_1_channel} and \eqref{user_2_channel} are single-dimensional MISO channels with $\textbf{n}_{I}, \textbf{n}_{Q} \sim \mathcal{N} \left(0, \frac{1}{2}\textbf{I}_{T}\right)$. Hence, the average receive SNR in every dimension is $\rho$. Since the rate-1 ROD for $N_{t}$ antennas is information lossless for a single dimensional $N_{t} \times 1$ MIMO channel (Theorem \ref{rod_IL}), for $j = 1, 2$, the maximum mutual information for User-$j$ is 
\begin{equation*}
\frac{1}{2}E\left[ \mbox{log}_{2}\left( 1 + \frac{\rho}{N_{t}}\textbf{h}_{i}\textbf{h}_{i}^{H}\right)\right].
\end{equation*}
Therefore, with overloading of notations, the sum capacity of the proposed scheme is given by,
\begin{equation}
\label{cap_rod_scheme}
E\left[\mbox{log}_{2}\left( 1 + \frac{\rho}{N_{t}}\textbf{h}_{2}\textbf{h}_{2}^{H}\right)\right]
\end{equation}
which is equal to the capacity of a $N_{t} \times 1$ collocated MIMO channel for an average SNR value of $\rho$.
However, the sum capacity of a $N_{t}$-MIMO-MAC is given in \eqref{sum_capacity_fading _mac_channel} which is equal to the capacity of a $2N_{t} \times 1$ MIMO channel for an average SNR value of of $\rho$.
By comparing \eqref{cap_rod_scheme} with $C(N_{t}, N_{t}, 1, p)$, it is not clear whether the proposed scheme is information lossless or information lossy for a $N_{t}$-MIMO-MAC $\forall N_{t}$. Through simulations, in Fig. \ref{rod_capacity_plot} (shown at the top of the next page), the sum capacity of the proposed scheme is compared with $C(N_{t}, N_{t}, 1, p)$ for $N_{t}$ = 2,  $N_{t}$ = 4 and  $N_{t}$ = 8 respectively at different SNR values. Note that when $N_{t}$ = 2 and 4, the proposed scheme is information lossy by a small margin and the difference in the capacity keeps diminishing as $N_{t}$ increases (See Fig. \ref{rod_capacity_plot} for $N_{t}$ = 8). In particular, using strong law of large numbers, for large values of $N_{t}$, we have
\begin{equation*}
\lim_{N_{t} \to \infty} E\left[\mbox{log}_{2}\left( 1 + \frac{\rho}{N_{t}}\textbf{h}\textbf{h}^{H}\right)\right] = C(N_{t}, N_{t}, 1, p). 
\end{equation*}
and hence the proposed designs are information lossless for large values of $N_{t}$. The above discussion can be summarized in the following theorem,
\begin{theorem}
\label{IL_theorem_rod}
For large values of $N_{t}$, STBC pairs from rate-1 RODs are information lossless for a $N_{t}$-MIMO-MAC.
\end{theorem}
\begin{figure*}
\centering
\includegraphics[width=7in]{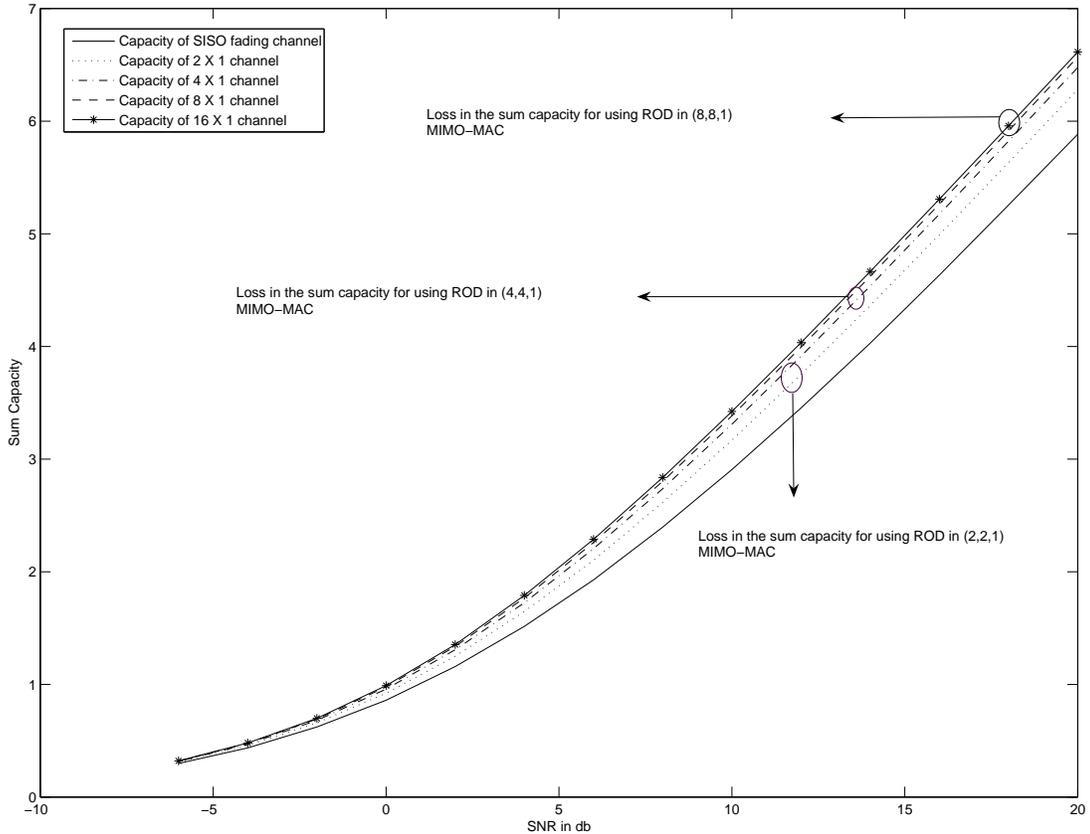}
\caption{Sum capacity of $N_{t}$-MIMO-MAC with RODs in comparison with $C(N_{t}, N_{t}, 1)$ for $N_{t} = 2, 4 ~\mbox{and}~ 8.$} 
\label{rod_capacity_plot}
\end{figure*}
\subsubsection{Minimum decoding complexity}
\label{sec4_subsec3_subsubsec2}
Apart from having the information lossless property for large values of $N_{t}$, the proposed codes also have the single-symbol ML decodable property. From \eqref{user_1_channel} and \eqref{user_2_channel}, the ML-decoding metrics for User-1 and User-2 are respectively given by
\begin{equation*}
\hat{\textbf{X}} = \mbox{arg} \min_{\mathcal{C}_{1}}||\textbf{y} - \sqrt{\frac{\alpha}{2N_{t}}}\textbf{X}\tilde{\textbf{h}}_{1}||^{2},
\end{equation*}
\begin{equation*}
\hat{\textbf{Y}} = \mbox{arg} \min_{\mathcal{C}_{2}}||\textbf{y} - \sqrt{\frac{\alpha}{2N_{t}}}\textbf{Y}\tilde{\textbf{h}}_{2}||^{2}.
\end{equation*}
Since $\tilde{\textbf{h}}_{j}$ are real vectors, and the designs, $\textbf{X}$ and $\textbf{Y}$ are RODs, for each user, every symbol can be decoded independent of the rest of the symbols. For more details on decoding the class of STBCs from RODs, we refer the reader to \cite{TJC}, \cite{Xl}. To the best of our knowledge, this is the first paper that addresses the design of STBC pairs with single symbol decodable property for two-user MIMO-MAC.
\begin{figure}
\centering
\includegraphics[width=3in]{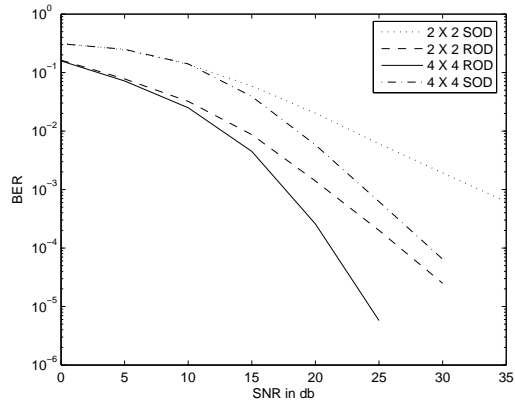}
\caption{BER comparison of STBC pairs from RODs and SODs} 
\label{ber_curves_saod_rod}
\end{figure}
\subsection{Simulation Results}
\label{sec4_subsec4}
In this subsection, we provide simulation results for the performance comparison of STBC pairs from SODs and RODs for a $N_{t}$-MIMO-MAC when $N_{t}$ = $2$ and $4$. We have used the Bit Error Rate (BER) which corresponds to errors in decoding the bits of both the users as error events of interest. For $N_{t} = 2$ and $4$, the rate-1 SOD pairs have been used for simulations. For both the cases the variables of the two users take values from a $4$-QAM signal set (with average energy per symbol being unity). With this, each user transmits $2$ bits per channel use (bpcu). For the second class of STBC pairs, RODs for $2$ and $4$ antennas are used wherein the variables of the design employed by User-1 take values from the 4-PAM signal set, $\sqrt(\frac{1}{5})\left\lbrace  -3, -1, 1, 3\right\rbrace$ whereas for User-2, the variables take values from the set $\sqrt(\frac{1}{5})\left\lbrace  -3i, -1i, 1i, 3i\right\rbrace$. With this, the transmission rate of 2 bpcu-per user is maintained for both the class of codes. For every codeword use, the destination has to decode for 8 bits (4 bits of each user) and 16 bits (8 bits for each user) for $N_{t} = 2$ and $N_{t} = 4$ respectively. BER comparison of the two schemes using the above designs is shown in Fig. \ref{ber_curves_saod_rod} where the plots show that STBC pairs from RODs perform better than the codes from SODs for both $N_{t} = 2$ and $4$. An intuitive reasoning for the above behaviour is that for the class of STBC pairs from RODs, there is no interference among the users. For the STBC pairs based on RODs, each real symbol is decoded in $\mathbb{R}$ whereas for the codes from SODs, decoding is in $\mathbb{R}^{4}$ for $N_{t} = 2$ and $\mathbb{R}^{8}$ for $N_{t} = 4$. 
\section{Discussion}
\label{sec5}
We have computed CC capacity regions of a two-user GMAC and proposed TCM schemes with the class of $M$-PSK signal sets and $M$-PAM signal sets. We have studied designing STBC pairs with low ML decoding complexity for a two-user MIMO-MAC with $N_{t}$ antennas at both the users and single antenna at the destination with the assumption of CSIT-P. Some possible directions for future work are as follows:  
\begin{itemize}
\item As a generalisation to this work, CC capacity/capacity regions for general multi terminal networks needs to be computed since in practice, communication takes place only with finite input alphabets. Also, design of coding schemes achieving rate tuples close to the CC capacity of general multi terminal networks is essential.
\item The set partitioning result presented in this paper can be generalized to the class of $M$-QAM alphabets.
\item In this paper, we have assumed equal average power constraint for both the users. It is clear that if unequal power constraint is considered, then the UD property is naturally attained. For such a setup, optimal labelling rules on the individual trellis has to be proved depending on the ratio of the average power constraints of the two users. It is straightforward to show that when the ratio of the average power constraints of the two users is sufficiently large, then irrespective of the relative angle of rotation between the alphabets, labelling of the individual trellises based on Ungerboeck partitioning is optimal in the sense of maximizing the criteria considered in this paper.
\item For a two-user GMAC, it has been shown that trellis code pairs based on TCM with $M$-PAM alphabet pairs significantly reduce the ML decoding complexity at the destination compared to TCM schemes with complex alphabet pairs. For a $K$-user GMAC with $K > 2$, designing coding schemes with low ML decoding complexity is an interesting direction of future-work.
\item In Section \ref{sec4}, we considered designing STBC pairs with low ML decoding complexity for a two-user MIMO-MAC with $N_{t}$ antennas at both the users and single antenna at the destination with the assumption of CSIT-P. Note that the assumption of CSIT-P has been exploited to obtain STBC pairs with low ML decoding complexity property. However, when the destination has multiple antennas, every transmit antenna of each user views more than one fading channel and hence phase compensation by the users is not possible. Therefore, design of low decoding complexity STBC pairs for such a set-up is not straightforward. In particular, design of low complexity STBC pairs for a MIMO-MAC without the assumption of CSIT-P is challenging.
\end{itemize}
\section*{Acknowledgment}
This work was partly supported by the DRDO-IISc Program on Advanced Research in Mathematical Engineering.

\begin{appendices}
\section{Proof of Proposition \ref{prop6}}
\label{proof_prop_6}
For $a, b \in \mathbb{C}$, let $\textit{l}(a, b)$ denote the line segment joining $a$ and $b$ in $\mathbb{R}^{2}$. It is to be noted that the complex points $0, r(I^{q-1})$ and $r(I^{q-1})e^{i\frac{2\pi}{M}}$ form the three vertices's of an isosceles triangle in $\mathbb{R}^{2}$. Since $r(O^{q}) \leq r(I^{q-1}),$ we have $d(0, r(O^{q})e^{i\frac{2\pi}{M}}) \leq  d(0, r(I^{q-1})e^{i\frac{2\pi}{M}})$. Therefore, the four points $r(O^{q}), r(O^{q})e^{i\frac{2\pi}{M}}, r(I^{q-1})$ and  $r(I^{q-1})e^{i\frac{2\pi}{M}}$ form the vertices's of an isosceles trapezoid $\Upsilon$ such that $\textit{l}(r(O^{q}), r(O^{q})e^{i\frac{2\pi}{M}})$ is parallel to $\textit{l}(r(I^{q-1}), r(I^{q-1})e^{i\frac{2\pi}{M}})$. Also, note that $d(r(I^{q-1}),  r(O^{q})e^{i\frac{2\pi}{M}})$ is the length of the diagonal of the trapezoid $\Upsilon$. Since the angle between the line segments $\textit{l}(r(O^{q}), r(O^{q})e^{i\frac{2\pi}{M}})$ and $\textit{l}(r(O^{q}), r(I^{q-1}))$ is obtuse, $d(r(I^{q-1}),  r(O^{q})e^{i\frac{2\pi}{M}}) \geq d(r(O^{q-1}), r(O^{q-1})e^{i\frac{2\pi}{M}})$. This completes the proof.
\section{Proof of Proposition \ref{prop7}}
\label{proof_prop_7}
We prove the inequality $2r(O^{M/2-1})\mbox{sin}(\frac{2\pi}{M}) \leq d(r(O^{q}), r(O^{q})e^{i\frac{2\pi}{M}})$ which can be written the following ratio : 
\begin{equation} 
\label{exp2}
\frac{2r(O^{M/2-1})\mbox{sin}(\frac{2\pi}{M})}{d(r(O^{q}), r(O^{q})e^{i\frac{2\pi}{M}})} = \frac{\mbox{sin}(\frac{\pi}{M} - \frac{\theta}{2})\mbox{sin}(\frac{2\pi}{M})}{\mbox{cos}(\frac{\theta}{2} + \frac{\pi q}{M})\mbox{sin}(\frac{\pi}{M})}.
\end{equation}
Since $\theta \leq \frac{pi}{M}$ and $q \leq M/2 - 3$, $\frac{\theta}{2} + \frac{\pi q}{M} \leq \frac{pi}{2} - \frac{5\pi}{M} < \frac{\pi}{2}$. Hence for all values of $\theta$ and $q$, $\mbox{cos}(\frac{\theta}{2} + \frac{\pi q}{M}) \geq \mbox{cos}(\frac{pi}{2} - \frac{5\pi}{M}) = \mbox{sin}(\frac{5\pi}{M})$. Also, $\mbox{sin}(\frac{\pi}{N} - \frac{\theta}{2}) \leq \mbox{sin}(\frac{\pi}{N})$. Therefore, the ratio in \eqref{exp2} satisfies the following inequality for $M \geq 8$, 
\begin{equation*}
\frac{\mbox{sin}(\frac{\pi}{M} - \frac{\theta}{2})\mbox{sin}(\frac{2\pi}{M})}{\mbox{cos}(\frac{\theta}{2} + \frac{\pi q}{M})\mbox{sin}(\frac{\pi}{M})} \leq \frac{\mbox{sin}(\frac{2\pi}{M})}{\mbox{sin}(\frac{5\pi}{2M})} \leq 1.
\end{equation*}
This completes the proof.
\section{Proof of Lemma 1}
\label{proof_lemma_1}
Since the structure of $\mathcal{S}^{eo}_{sum}$ and $\mathcal{S}^{oe}_{sum}$ are identical, we find the minimum distance of $\mathcal{S}^{eo}_{sum}$. Since the points of $\mathcal{S}^{eo}_{sum}$ are maximally separated on every circle and $O^{M/2-1}$ is the innermost circle, $d(r(O^{M/2-1}), r(O^{M/2-1})e^{i\frac{4\pi}{M}}) = 2r(O^{M/2-1})\mbox{sin}(\frac{2\pi}{M}) = d_{1}$ is a contender for $d^{eo}_{min}$. For this to be true, it is to be shown that all other intra-distances in the set are larger than or equal to $d_{1}$. In particular, the distances between the points on any two consecutive circles must be larger than $d_{1}$.
Firstly, it is shown that a point on $I^{q}$ and a point on $O^{q-1}$ which have an angular separation of $0$ radians are separated by a distance larger than $d_{1}$ for all $q = 2$ to $M/2 - 2$. In that direction, the first observation is that $r(O^{1}) - r(I^{2}) = d_{1}$. From the results of the Proposition \ref{prop3}, $ r(O^{k}) - r(I^{k+1}) \geq d_{1}$ for all $k \geq 2$. Hence the points on $I^{q}$ and $O^{q-1}$ are separated by a distance larger than $d_{1}$ for all $q = 2$ to $M/2 - 2$.\\
\indent Secondly, it is to be verified if the point on $O^{q}$ and the point on $I^{q-1}$ are separated by a distance larger than $d_{1}$ for all $q = 1$ to $M/2 - 1$. It is shown that the points on $O^{q}$ and $I^{q-1}$ having an angular separation of $\frac{2\pi}{M}$ are separated by a distance larger than $d_{1}$ only for $q = 1$ to $M/2 - 3$ but not for $q = M/2 - 1$. For $q = M/2 - 1$, $d(r(I^{q-1}),  r(O^{q})e^{i\frac{2\pi}{M}})$ can be lesser than $2r(O^{M/2-1})\mbox{sin}(\frac{2\pi}{M})$ for certain values of $\theta$. Therefore, we prove $d(r(I^{q-1}),  r(O^{q})e^{i\frac{2\pi}{M}}) \geq 2r(O^{M/2-1})\mbox{sin}(\frac{2\pi}{M})$ only for $q = 1$ to $M/2 - 3$ using the following sequence of inequalities, 

{\footnotesize
\begin{equation*}
d(r(I^{q-1}),  r(O^{q})e^{i\frac{2\pi}{M}}) \geq d(r(O^{q}), r(O^{q})e^{i\frac{2\pi}{M}}) \geq 2r(O^{M/2-1})\mbox{sin}(\frac{2\pi}{M}).
\end{equation*}
}

The first inequality is proved in Proposition \ref{prop6} whereas the second inequality is proved in Proposition \ref{prop7}. Hence, the points on $O^{q}$ and $I^{q-1}$ are separated by a distance larger than $d_{1}$ for all $q = 1$ to $M/2 - 3$. 
Therefore, $d^{eo}_{min} = \mbox{min}\left(d(r(I^{q-1}),  r(O^{q})e^{i\frac{2\pi}{M}}), d_{1}\right) $.
This completes the proof.
\section{Proof of Lemma 2}
\label{proof_lemma_2}
Since the structure of $\mathcal{S}^{ee}_{sum}$ and $\mathcal{S}^{oo}_{sum}$ are the same, we find the minimum distance of $\mathcal{S}^{ee}_{sum}$ only. Since the points of $\mathcal{S}^{ee}_{sum}$ are maximally separated with $\phi_{min} = \frac{4\pi}{M}$ on every circle and $I^{M/2-1}$ is the innermost circle, $d(r(I^{M/2-1}), r(I^{M/2-1})e^{i\frac{4\pi}{M}}) = 2r(I^{M/2-1})\mbox{sin}(\frac{2\pi}{M}) = d_{2}$
is a contender $d^{ee}_{min}$. For this to be true, it is to be shown that the distances between the points on any two consecutive circles must be larger than $d_{2}$. We show that a point on $O^{q}$ and a point on $I^{q-1}$ which have an angular separation of $0$ radians are separated by a distance larger than $d_{2}$ for all $q = 2$ to $M/2 - 2$. In that direction, the first observation is that $r(I^{1}) - r(O^{2}) = d_{2}$. From the result of the Proposition \ref{prop4} in Section \ref{sec1}, $r(I^{k}) - r(O^{k+1}) \geq d_{2}$ for all $k \geq 2$. Hence the points on $O^{q}$ and $I^{q-1}$ are separated by a distance larger than $d_{2}$ for all $q = 2$ to $M/2 - 2$ .\\
\indent Secondly, it is shown that the point on $I^{q}$ and the point on $O^{q-1}$ are separated by a distance larger than $d_{2}$ for all $q = 1$ to $M/2 - 1$. i.e. we prove $d(r(O^{q-1}),  r(I^{q})e^{i\frac{2\pi}{M}}) \geq d_{2}$ for $q = 1$ to $M/2 - 1$. In that direction, for $q = 1$ to $M/2 - 3$, we show that

{\footnotesize 
\begin{equation}
\label{eq1_sec3}
d(r(O^{q-1}),  r(I^{q})e^{i\frac{2\pi}{M}}) \geq d(r(I^{q}), r(I^{q})e^{i\frac{2\pi}{M}}) \geq 2r(I^{M/2-1})\mbox{sin}(\frac{2\pi}{M}).
\end{equation}
}

For the case when $q = M/2 - 1$, we show that 

{\footnotesize 
\begin{equation}
\label{eq2_sec3}
d(r(O^{q-1}),  r(I^{q})e^{i\frac{2\pi}{M}}) \geq d(r(O^{q-1}), r(I^{q})) \geq 2r(I^{M/2-1})\mbox{sin}(\frac{2\pi}{M}).
\end{equation}
}

The proof for the first lower bounds of \eqref{eq1_sec3} and \eqref{eq2_sec3} are on the similar lines of the proof for Proposition \ref{prop6}. The proofs of the second lower bounds of \eqref{eq1_sec3} and \eqref{eq2_sec3} are in Proposition \ref{prop8} and Proposition \ref{prop9} respectively.
Therefore, $d^{ee}_{min} = d_{2} = 4\mbox{sin}\left(\frac{\theta}{2})\mbox{sin}(\frac{2\pi}{M}\right)$.
This completes the proof.
\section{Proof of Theorem \ref{rod_IL}}
\label{proof_theorem_rod_IL}
Let $\textbf{X}$ represents the $l \times N_{t}$ ROD for $N_{t}$ antennas in the variables $x_{1}, x_{2} \cdots x_{l}$. Note that the number of channel uses is equal to the number of real variables since $\textbf{X}$ is a rate-1 ROD. Also, $\textbf{X}$ has the following column vector representation,
\begin{equation*}
\textbf{X} = \left[\textbf{A}_{1}\textbf{x} ~\textbf{A}_{2}\textbf{x} ~\cdots ~\textbf{A}_{N_{t}}\textbf{x}\right] 
\end{equation*}
where $\left\lbrace \textbf{A}_{i} ~|~ i = 1 \mbox{ to } N_{t}\right\rbrace $ is the set of column vector representation matrices of $\textbf{X}$ and $\textbf{x}^{T} = \left[x_{1} ~x_{2} ~\cdots ~x_{l} \right]$. The MISO channel equation with the above design is, $\textbf{y} = \sqrt{\frac{\rho}{N_{t}}}\textbf{X}\textbf{h} + \textbf{n}$ where $\rho$ is the average receive SNR and $\textbf{n} \sim \mathcal{N} \left(0, 1\right)$. The above channel equation can also be written as
\begin{equation*}
\textbf{y} = \sqrt{\frac{\rho}{N_{t}}}\hat{\textbf{H}}\textbf{x} + \textbf{n}
\end{equation*}
where $\hat{\textbf{H}} =  \sum_{i = 1}^{N_t} h_{i}\textbf{A}_{i}$. If the channel from every antenna to the destination is i.i.d Rayleigh distributed with unit mean, the capacity of the above channel is
\begin{equation*}
\frac{1}{2}E~\left[\mbox{log}_{2}\left(\mbox{det}\left( \textbf{I}_{l} + \frac{\rho}{N_{t}}\hat{\textbf{H}}\hat{\textbf{H}}^{H}\right)\right)\right].
\end{equation*}
Since $\textbf{A}_{i}$'s are unitary and $\textbf{A}_{i}\textbf{A}_{j}^{T} + \textbf{A}_{j}\textbf{A}_{i}^{T} = \textbf{0}_{T \times T}$ $\forall ~i, j$ such that $i \neq j$, we have $\hat{\textbf{H}}\hat{\textbf{H}}^{H} = \left(\sum_{i = 1}^{N_t} h_{i}^{2}\right)\textbf{I}_{l}$ and hence the capacity of a single-dimensional MISO channel with the ROD, $\textbf{X}$ is
\begin{equation*}
\frac{1}{2}E~\left[\mbox{log}_{2}\left(1 + \frac{\rho}{N_{t}}(\sum_{i = 1}^{N_t} h_{i}^{2})\right)\right].  
\end{equation*}
This completes the proof.
\end{appendices}
\end{document}